\documentclass[11pt]{article}
\usepackage[T1]{fontenc}
\usepackage[margin=1in]{geometry}
\usepackage{graphicx} 

\usepackage{multicol}
\usepackage{makecell}
\usepackage{array}
\usepackage{amssymb}
\usepackage{amsmath}
\usepackage{accents}
\usepackage{mathtools}
\usepackage{mathrsfs}
\usepackage{bm}
\usepackage{complexity}
\usepackage{xcolor}
\usepackage{graphicx}
\usepackage{subcaption}
\usepackage{floatrow}
\newfloatcommand{capbtabbox}{table}[][\FBwidth]
\usepackage{soul}
\usepackage[ruled, linesnumbered, noend]{algorithm2e}
\usepackage{csquotes}
\usepackage{tikz}
\usepackage{tikz-3dplot}
\usetikzlibrary{decorations.pathreplacing}
\usetikzlibrary{positioning}
\usetikzlibrary{calc}
\usetikzlibrary{patterns}
\usetikzlibrary{backgrounds}
\usetikzlibrary{fit}
\usetikzlibrary{arrows.meta}
\usetikzlibrary{matrix}
\usetikzlibrary{3d}
\usetikzlibrary{decorations.pathmorphing}
\pgfdeclarelayer{bg}    
\pgfsetlayers{bg,main}  

\usepackage{amsthm}
\usepackage{thmtools,thm-restate}
\newtheorem{mresult}{Main Result}
\newtheorem{theorem}{Theorem}[section]
\newtheorem*{claim}{Claim}

\newtheorem{lemma}[theorem]{Lemma}
\newtheorem{proposition}[theorem]{Proposition}
\newtheorem{definition}[theorem]{Definition}

\newtheorem{remark}[theorem]{Remark}

\newtheorem{example}[theorem]{Example}

\usepackage[
backend=biber,
style=alphabetic,
giveninits=true
]{biblatex}
\DeclareNameAlias{default}{family-given}
\newcommand{\univ}{\Sigma}
\newcommand{\glang}{\Lambda}

\newcommand{\overbar}[1]{\mkern 1.5mu\overline{\mkern-1.5mu#1\mkern-1.5mu}\mkern 1.5mu}

\newcommand{\cl}{\operatorname{cl}}

\newcommand{\gtest}{\accentset{\vee}{\rho}}

\newcommand{\conn}{\lambda_\mathrm{max}}

\DeclareMathOperator{\supp}{\textup{supp}}

\makeatletter
\providecommand{\bigsqcap}{%
	\mathop{%
		\mathpalette\@updown\bigsqcup
	}%
}
\newcommand*{\@updown}[2]{%
	\rotatebox[origin=c]{180}{$\m@th#1#2$}%
}
\makeatother

\makeatletter
\newcommand{\vast}{\bBigg@{4}}
\newcommand{\Vast}{\bBigg@{5}}
\makeatother

\DeclareMathOperator{\raisedrightarrow}{\raisebox{+0.15ex}{$\rightarrow$}}
\makeatletter
\newcommand{\rightarroweq}{\mathrel{\mathpalette\rightarroweq@\relax}}
\newcommand{\rightarroweq@}[2]{%
  \vtop{
      \sbox\z@{$\m@th#1\raisedrightarrow$}%
    \ialign{%
      \hfil##\hfil\cr
      \copy\z@\cr
      \noalign{\nointerlineskip\kern0.7\ht\z@}
      \makebox[\wd\z@][s]{$\m@th#1\relbar\hss\relbar$}\cr
      \noalign{\kern-0.7\ht\z@}
    }%
  }%
}
\makeatother

\DeclareMathOperator{\sraisedrightarrow}{\raisebox{+0.075ex}{$\scriptstyle\rightarrow$}}
\makeatletter
\newcommand{\srightarroweq}{\mathrel{\mathpalette\srightarroweq@\relax}}
\newcommand{\srightarroweq@}[2]{%
  \vtop{
      \sbox\z@{$\m@th#1\sraisedrightarrow$}%
    \ialign{%
      \hfil##\hfil\cr
      \copy\z@\cr
      \noalign{\nointerlineskip\kern0.7\ht\z@}
      \makebox[\wd\z@][s]{$\m@th#1\relbar\hss\relbar$}\cr
      \noalign{\kern-0.7\ht\z@}
    }%
  }%
}
\makeatother

\DeclareMathOperator{\sraisedleftarrow}{\raisebox{+0.075ex}{$\scriptstyle\leftarrow$}}
\makeatletter
\newcommand{\sleftarroweq}{\mathrel{\mathpalette\sleftarroweq@\relax}}
\newcommand{\sleftarroweq@}[2]{%
  \vtop{
      \sbox\z@{$\m@th#1\sraisedleftarrow$}%
    \ialign{%
      \hfil##\hfil\cr
      \copy\z@\cr
      \noalign{\nointerlineskip\kern0.7\ht\z@}
      \makebox[\wd\z@][s]{$\m@th#1\relbar\hss\relbar$}\cr
      \noalign{\kern-0.7\ht\z@}
    }%
  }%
}
\makeatother

\newcommand{\ideal}[1]{\mathop{\downarrow}#1}

\providecommand{\L}{}
\renewcommand{\L}{\mathscr{L}}

\newcommand{\setpt}[4]{%
  \coordinate (#1) at (#2,#3,#4);%
  \expandafter\xdef\csname X@#1\endcsname{#2}%
  \expandafter\xdef\csname Y@#1\endcsname{#3}%
  \expandafter\xdef\csname Z@#1\endcsname{#4}%
}

\newcommand{\proposal}[1]{\ignorespaces}

\title{Approximation Algorithms for Matroidal Prerequisite Systems}
\author{Robert P. Streit\thanks{Department of Electrical and Computer Engineering at The University of Texas at Austin.} \and Vijay K. Garg\footnotemark[1]\textsuperscript{\;\:,}\thanks{Partially supported by NSF CNS-1812349, and the Cullen Trust for
Higher Education Endowed Professorship.}}
\date{May 6, 2026}

\begin{document}

\maketitle

\begin{abstract}
    Optimal selections in a decision process are often constrained by prerequisites. However, such prerequisites can encode functional rather than literal dependencies, so a required dependency may be supplied by one or several interacting alternatives.
    We introduce \emph{matroidal prerequisite systems (MPS)}, a combinatorial constraint structure where a poset specifies prerequisites while a matroid determines when those prerequisites have been satisfied by its span.
    This creates an order-sensitive notion of feasibility over words, where feasible words are associated with independent sets, while dependencies may be fulfilled through substitutable functionality.

    Our main contribution is approximation algorithms for nonnegative additive maximization and monotone submodular maximization over the feasible words of an MPS.
    The guarantees are determined by two structural parameters, namely the maximum matroid rank $\Delta$ of a principal ideal in the poset and the maximum matroid connectivity $\conn$.
    These parameters measure the distance an MPS is from encoding a matroid or a poset antimatroid, respectively, both of which are generalized by an MPS.
    For additive maximization, we obtain efficient deterministic $\Delta$- and $(1+\conn)$-approximation algorithms.
    By extending these techniques, we obtain efficient deterministic $(2+\conn)$-approximation and randomized $(\Delta^2\cdot(1 - 1/e - \delta)^{-1})$-approximation algorithms for all $\delta >0$ for submodular maximization.
    The algorithm design and analysis use the theory of polymatroid greedoids, through a cryptomorphism we prove between an MPS and a \emph{strong polymatroid greedoid}.
    Finally, an approximation-preserving reduction from densest $k$-subgraph shows it is not possible to efficiently compute a $\min\{\Delta,\conn\}^{o(1)}$-approximation to nonnegative additive maximization over the feasible words of an MPS under the Gap Exponential Time Hypothesis.
    Thus, an MPS provides a tractable, but provably nontrivial, framework
for combinatorial optimization with interacting prerequisites, independence, and substitution.
\end{abstract}

\section{Introduction}\label{sec:intro}

In a discrete decision process, some decisions are not meaningful, feasible, or useful until other functionality is already present.
This is perhaps most obvious in task scheduling problems, where dependencies impose constraints on what can and cannot be completed.
However, such dependencies are rarely literal.
Some project task may require not an earlier task itself, but rather some capability, interface, or information that earlier tasks are meant to provide.
These required functionalities may be satisfied by alternative selections of one or many interacting decisions.
Yet, existing prerequisite models (e.g. \cite{picard1976maximal,johnson1983knapsacks,samphaiboon2000heuristic,mccormick2017primal,kolliopoulos2007partially}) typically require prerequisites to be selected explicitly; i.e. if $y$ is a prerequisite of $z$, then $z$ can only be chosen after $y$.
So, there is no awareness of the overlap or substitutable nature of the functionalities of the decisions.

\emph{Matroids} are set families encoding such a substitution structure.
Recall, a matroid $\mathcal{M}$ is a nonempty set family, whose members are called \emph{independent}, such that every subset of an independent set is also independent and for every independent $X$ and $Y$ with $|Y| > |X|$ there exists $z \in Y\setminus X$ such that $X+z$ is independent.
The \emph{rank} $r_\mathcal{M}(X)$ of an arbitrary set $X$ is the size of a maximal independent set it contains, while the \emph{span} $\sigma_\mathcal{M}(X)$ is the union of $X$ with the elements which do not increase its rank.
Now, if feasible sequences of decisions $x_1\ldots x_\ell$ are always such that $\{x_1,\ldots,x_\ell\}$ is independent,
then the span $\sigma_\mathcal{M}(\{x_1,\ldots,x_\ell\})$ can encode the functionalities covered by $x_1\ldots x_\ell$.
So, two decisions $y,z\notin\sigma_\mathcal{M}(\{x_1,\ldots,x_\ell\})$ which are not made redundant by these functionalities (in the sense they are spanned by $\{x_1\ldots x_\ell\}$) are \emph{$\{x_1,\ldots,x_\ell\}$-substitutable} if and only if $\sigma_\mathcal{M}(\{x_1,\ldots,x_\ell,y\}) = \sigma_\mathcal{M}(\{x_1,\ldots,x_\ell,z\})$.
This is, given what is already covered by $x_1\ldots x_\ell$, any new functionality added by $y$ or $z$ leads to equivalent coverage.
Matroids hold a central place in discrete optimization not only because their structure models many problems, but also gives algorithms for many objective functions \cite{rado1957note,edmonds1971matroids,nemhauser1978analysis,calinescu2011maximizing,filmus2014monotone,sviridenko2017optimal}.

In this work, we study the optimization problems that arise when prerequisite and matroidal substitution structure are forced to interact.
Our reasons are both modeling based and algorithmic.
Many applications lead to prerequisite structure defined via substitutable functionalities rather than irreplaceable objects, and we would like to capture this.
We will also see the resulting constraint systems form a structured class weaker than matroids that still admits approximation algorithms for additive and submodular maximization.
To that end, we propose the \emph{matroidal prerequisite system (MPS)}, which is a simple language (words without repeating letters) over a finite alphabet $\univ$ defined by an interaction between a matroid $(\univ,\mathcal{M})$ and poset $(\univ,\leadsto)$.
The words $y_1\ldots y_\ell$ in the language, denoted by $\mathcal{L}(\mathcal{M},\leadsto)$, are such that their letters always form independent sets, and each $y_i$ is such that all prerequisites $x \leadsto y_i$ are either spanned by the prefix $\{y_1,\ldots, y_{i-1}\}$ in the matroid or $x$ and $y_i$ are $\{y_1,\ldots,y_{i-1}\}$-substitutable.
Examples are presented in Section~\ref{sec:mps}.
The upshot is the poset and matroid retain their usual roles, but no
longer act independently.
Specifically, the poset specifies which decisions require earlier prerequisite coverage, and the matroid determines when that prerequisite coverage has been supplied via its span while continuing to enforce nonredundancy via independence.
An MPS generalizes matroids, as $\mathcal{L}(\mathcal{M},\leadsto)$ is the permutations of independent sets if the ordering relation is empty, as well as \emph{poset antimatroids} (the language of linear extensions of a poset) as $\mathcal{L}(\mathcal{M},\leadsto)$ is the linear extensions of $(\univ,\leadsto)$ if $\mathcal{M}$ is free (all sets are independent).

Our main contributions are approximation algorithms for maximization over the words of an MPS of nonnegative additive and monotone submodular objective functions.
Our guarantees are given in terms of structural parameters that intuitively measure the distance of an MPS from being a matroid or poset antimatroid.
The first is the largest matroid rank of a principal ideal in $(\univ,\leadsto)$,
\begin{equation}\label{eq:delta-def}
    \Delta\triangleq \max_{y\in\univ} r_\mathcal{M}(\ideal y).
\end{equation}
If the MPS is a matroid then each principal ideal is size one, and so $\Delta = 1$.
As the prerequisite structure becomes more complex, $\Delta$ will increase.
The second parameter is defined by the largest value taken by the (Tutte) connectivity function $\lambda$ (see \cite{tutte1966connectivity,welsh2010matroid,oxley2006matroid}) of the matroid,
\begin{equation}\label{eq:conn-def}
    \conn \triangleq \max_{X\subseteq\univ} \lambda(X),\quad\text{where }\lambda(X)\triangleq r_\mathcal{M}(X) + r_\mathcal{M}(\univ\setminus X) - r_\mathcal{M}(\univ).
\end{equation}
The connectivity $\lambda(X)$ computes how much the complement $\univ\setminus X$ can cover $X$ in its span, in terms of rank.
Equivalently, $\lambda(X)$ measures how much of $X$ can be substituted with its complement.
So, $\conn$ measures how complex the substitutes structure is.
Note $\conn \geq 0$.
For free matroids where every set is independent, we see $\conn =0$.
So, $\conn=0$ when the MPS defines a poset antimatroid and increases as the substitutes structure becomes more complex.
This leads to our first main result.

\begin{mresult}[Theorems~\ref{thm:scaffold-approx-add},~\ref{thm:scaffold-approx-sub},~\ref{thm:mag-approx}, and~\ref{thm:mag-approx-sub}]\label{mr:algorithm}
    Let $\mathcal{L}(\mathcal{M},\leadsto)$ be a matroidal prerequisite system.
    \begin{enumerate}
        \item There exists an efficient deterministic $\min\{\Delta,1 + \conn\}$-approximation algorithm for nonnegative additive maximization over the feasible words of $\mathcal{L}(\mathcal{M},\leadsto)$,
        \item There exists an efficient deterministic $(2 + \conn)$-approximation and, for every $\delta > 0$, an efficient randomized $(\Delta^2\cdot(1 - 1/e - \delta)^{-1})$-approximation algorithm for monotone submodular maximization over the feasible words of $\mathcal{L}(\mathcal{M},\leadsto)$.
    \end{enumerate}
\end{mresult}

Observe, our methods achieve exact solutions for additive maximization whenever the MPS is a matroid or poset antimatroid.
These are the two ``easy'' cases, and the approximation smoothly degrades as we drift towards a more complex MPS structure (as $\Delta$ and $\conn$ will increase).
For submodular maximization we achieve essentially the best possible approximation \cite{nemhauser1978best} when the MPS is a matroid.
For poset antimatroids, the problem becomes trivialized as the entirety of $\univ$ can always be permuted into a linear extension, so outputting this permutation always leads to an exact solution even for monotone submodular objectives.
But, our guarantee would give a $2$-approximation.
Whether it is possible to improve this to smoothly interpolate between this scenario and more complex MPS structures (like in the additive case) is something we leave to future work.
Finally, we note that the expected $(\Delta^2\cdot(1 - 1/e - \delta)^{-1})$-approximation guarantee comes from using the continuous greedy algorithm \cite{calinescu2011maximizing} as a black-box subroutine.
We present our methods with this chosen subroutine because it is the most recognizable to the widest audience, but alternative methods \cite{sviridenko2017optimal,filmus2014monotone,buchbinder2025deterministic} can be invoked to obtain certain improvements.
See Remark~\ref{rem:imp}.

Our algorithmic techniques make use of \emph{strong polymatroid greedoids} recently introduced in \cite{streit2024polymatroid}.
Specifically, a polymatroid greedoid is a greedoid whose feasibility structure can be encoded by a polymatroid rank function called a \emph{representation}.
Strong polymatroid greedoids are a subclass admitting certain closure operators and a canonical representation called \emph{the greatest representation}.
We prove a cryptomorphism in Theorem~\ref{thm:spg-iff-mps} showing a greedoid is a strong polymatroid greedoid if and only if it can be encoded as an MPS.
This allows us to use the greatest representation in our algorithms.
In particular, Section~\ref{sec:scaffold} presents algorithms based off of solving a relaxation over the base polytope of the greatest representation, while Section~\ref{sec:mag} uses the greatest representation to create modified greedy algorithms.
The correctness of the algorithmic results will follow from Lemma~\ref{lem:projection} in Section~\ref{subsec:fill}, which discusses how to iteratively ``fill-in'' sets of letters contained in feasible words to converge to a full feasible word in a strong polymatroid greedoid.

To complement the algorithmic results, we also investigate the hardness of nonnegative additive maximization.
We give an approximation-preserving reduction from Densest $k$-Subgraph.
By \cite{manurangsi2017almost}, this shows our algorithmic guarantees are meaningful in that exponential improvements are not possible under the Gap Exponential Time Hypothesis (see \cite{manurangsi2017birthday,dinur2016mildly,chalermsook2017gap}).
\begin{mresult}\label{mr:complexity}
    Assuming the Gap Exponential Time Hypothesis, there exists no $\min\{\Delta,\conn\}^{o(1)}$-approximation algorithm for additive maximization over a matroidal prerequisite system.
\end{mresult}

If one only assumes the Exponential Time Hypothesis, then weaker inapproximability results are still achieved; see Remark~\ref{rem:eth}.
While Main Result~\ref{mr:complexity} shows constant or logarithmic-approximation algorithms in these parameters are unlikely, the result doesn't comment on sublinear approximations in their entirety.
We leave the resolution of this gap in the additive case, as well as an investigation of the approximation hardness with more general submodular objective functions, to future work.

\paragraph{Overview}
We begin by surveying related work and making comparisons in Section~\ref{sec:related-work}.
Then, in Section~\ref{sec:preliminaries} we overview some background.
Section~\ref{sec:mps} will introduce the MPS, and discuss examples and applications.
Section~\ref{subsec:mps-greedoid} then examines the greedoid theory of an MPS, and proves a cryptomorphism between the class of MPS and strong polymatroid greedoids from \cite{streit2024polymatroid} in Theorem~\ref{thm:spg-iff-mps}.
Section~\ref{subsec:mps-gr} shows that the greatest representation of a MPS can be evaluated efficiently via a matroid rank query, while Section~\ref{subsec:fill} proves that every sequence of letters such that the marginal return of each letter against its prefix in the greatest representation is greater than zero can be ``filled'' into a feasible word by repeated letter insertions.
This is Lemma~\ref{lem:projection} to be precise.

In Section~\ref{sec:scaffold} we design approximation algorithms with guarantees varying in $\Delta$.
Section~\ref{subsec:scaffold-am} gives a method for additive maximization which is based on solving a relaxation defined by the greatest representation and projecting the solution onto a feasible word using Lemma~\ref{lem:projection}.
This is Algorithm~\ref{alg:greedy-scaffold}, which Theorem~\ref{thm:scaffold-approx-add} shows is a $\Delta$-approximation algorithm.
Then, in Section~\ref{subsec:scaffold-sub} we obtain a method for monotone submodular maximization.
This is based on applying the continuous greedy algorithm \cite{calinescu2011maximizing} to the \emph{natural matroid} construction \cite{helgason1974aspects,bonin2023natural} of the greatest representation and using the solution to define a set of weights underestimating the objective function,
which is fed into our additive maximization scheme in Algorithm~\ref{alg:greedy-scaffold} to obtain a feasible word.
The analysis in Theorem~\ref{thm:scaffold-approx-sub} shows this process leads to an expected approximation at most a $\Delta^2$ factor worse than that of \cite{calinescu2011maximizing} for submodular maximization over matroids.
In Section~\ref{sec:mag}, we then design approximation algorithms with guarantees varying in $\conn$.
Section~\ref{sec:mag-am} begins by examining algorithms for additive objectives, and presents a modified greedy algorithm (Algorithm~\ref{alg:mag}) using the greatest representation to adjust the perceived value of available letters in each round, and the closure operator of strong polymatroid greedoids to determine where to insert selected letters into a working solution.
The correctness is again witnessed by Lemma~\ref{lem:projection}, and Theorem~\ref{thm:mag-approx} shows the resulting solution is a $(1+\conn)$-approximation.
Then, Section~\ref{subsec:mag-sub} extends the technique to accommodate submodular objective functions.
This gives Theorem~\ref{thm:mag-approx-sub}, stating Algorithm~\ref{alg:mag-sub} obtains a $(2+\conn)$-approximation.
Finally, we conclude by proving Main Result~\ref{mr:complexity} in Section~\ref{sec:hardness}.

\section{Related Work}\label{sec:related-work}

Our MPS is in many ways a relaxation of matroids for designing approximation algorithms in more general settings.
Other approaches have been made in this spirit.
See for example $p$-systems \cite{nemhauser1978analysis}, $k$-intersection \cite{korte1978analysis,lee2010submodular}, $k$-extendible \cite{mestre2006greedy}, and $k$-exchange \cite{feldman2011improved} systems.
All come with approximate guarantees decaying in the ``distance'' from a matroid, like the $\Delta$ parameter used in Main Result~\ref{mr:algorithm}.
But, a key difference is our feasibility model is order-sensitive.
This is to say, where previous works maintain a subset-closed set family, our approach injects a notion of prerequisites using a poset that destroys such structure.
This is why we interpret our MPS as a greedoid \cite{korte2012greedoids}, as this theory accommodates order-sensitive feasibility.

On that note, there do exist works mixing prerequisite structure with other combinatorial constraint systems.
A notable example is precedence constrained knapsack variants \cite{samphaiboon2000heuristic,johnson1983knapsacks,kolliopoulos2007partially} and maximum closure \cite{picard1976maximal,mccormick2017primal} problems.
However, the prerequisite constraints of these problems require a ``complete'' closure, in the sense that an element can only be selected once \emph{all} of its prerequisites are present.
In contrast, we use a matroid to encode notion of substitution or redundancy.
Along these lines, AND/OR scheduling problems \cite{mohring2004scheduling} use Boolean formulae to allow for alternatives through disjunction.
However, in \cite{mohring2004scheduling} the formulae are explicit inputs, while our delegation to the matroid induces a stateful and implicit notion of substitutability which varies over the decision process.
As such, our approaches are not directly comparable.

Our work uses greedoids, a structure introduced by Korte and Lovász to generalize matroids by removing the definitional dependence on subset-closed set families \cite{korte1984structural,korte2012greedoids,bjorner1992introduction}.
For linear optimization, \cite{korte1984greedoids,goetschel1986linear} give a necessary and sufficient called \emph{strong exchange} for the correctness of the greedy algorithm in greedoids, which was later generalized to the idea of a \emph{matroid embedding} in \cite{helman1993exact}.
In Section~\ref{sec:mps} we will see that the greedy algorithm fails to solve additive maximization on an MPS, and so these results do not apply to our setting.
Without strong exchange, there do exist certain characterizations on the types of objective functions which the greedy algorithm can correctly solve \cite{korte2012greedoids,szeszler2022sufficient,faigle1985ordered}, but they are sensitive to the greedoid structure and quite restrictive.
Polymatroid greedoids are introduced in \cite{korte1985polymatroid}, and recently studied in \cite{streit2024polymatroid}.
This class is defined by an interaction between a polymatroid and a greedoid which allows for the use of submodular analysis (which is typically not fully available in the greedoid setting, unlike matroids).
The latter work \cite{streit2024polymatroid} introduces \emph{strong polymatroid greedoids}, which give special properties we use for the design and analysis of our algorithms after showing an MPS is cryptomorphically equivalent to such a greedoid.

The examination of submodular optimization or approximation algorithms for systems without subset-closure like greedoids appears overlooked in the literature.
An exception is \emph{distributive supermatroids} \cite{gottschalk2015submodular,maehara2022rank} (also known as poset matroids), which are also an interaction between posets and matroids.
However, a distributive supermatroid equals the intersection of a matroid and a poset antimatroid (see Ch. 9 of \cite{korte2012greedoids}), which we show cannot describe an MPS in Section~\ref{sec:mps}.

\section{Preliminaries}\label{sec:preliminaries}

In all concepts that follow we let the alphabet $\univ$ serve as a ground set.
Before proceeding, some conventions: For any set $X$, let $X!$ give its permutations, and $\chi_X \in \{0,1\}^\univ$ be its characteristic vector (i.e. $\chi_X(y) = 1$ if $y \in X$ and zero otherwise).
We often iterate over words written like $y_1\ldots y_{i-1}$.
So, if $i = 0$, then we are implicitly writing $y_1\ldots y_0$.
Take the convention that $y_1\ldots y_0 = \epsilon$, where $\epsilon$ is the empty word.
Similarly, let $\{y_1,\ldots,y_0\} = \varnothing$.
Finally, for a vector $v \in \mathbb{R}^\univ$, the \emph{support} $\supp(v) \triangleq \{y\in\univ\mid v(y) \neq 0\}$ is the set of labels for the nonzero coordinates in $v$.

\paragraph{Submodularity} An \emph{additive function} is a set function defined by assigning each letter a weight via $w:\univ \to\mathbb{R}$ and mapping sets to the sum of the weights of letters they contain, like $X\mapsto\sum_{y\in X}w(y)$.
A \emph{submodular function} $f:2^\univ\to\mathbb{R}$ is a set function satisfying the \emph{law of diminishing returns},
\[
    X\subseteq Y \implies (f/X)(z) \geq (f/Y)(z),\quad\forall z \in \univ,
\]
where $(f/X)(z) \triangleq f(X+z) - f(X)$ is the evaluation of the \emph{contraction of $f$ by $X$} at the letter $z$.
We adopt the convention that $(f/X)(z) = 0$ whenever $z \in X$.
Note, if $X = \{y_1,\ldots,y_\ell\}$, then, 
\[
    f(X) = \sum_{i = 1}^\ell (f/\{y_1,\ldots,y_{i-1}\})(y_i).
\]
And so, submodular functions generalize additive functions.
Nonnegative additive and monotone (i.e. $X\subseteq Y$ implies $f(X) \leq f(Y)$) submodular functions are the objective functions we consider.
We will assume that $f$ is \emph{normalized}, i.e. $f(\varnothing) = 0$, to remove any translation effects from our approximation ratios.
This is a common maneuver, see \cite{nemhauser1978analysis,calinescu2011maximizing,filmus2014monotone} for example.

\begin{algorithm}[t]
    \caption{Edmonds' Greedy Algorithm \cite{edmonds2003submodular}}\label{alg:edmonds}
    \DontPrintSemicolon 
    choose $y_1\ldots y_n \in \univ!$ so that
    $w(y_1) \ge w(y_2) \ge \ldots \ge w(y_n)$\tcp*{sort the coordinates}
    \;
    
    \For{$i \in \{1,\ldots,n\}$}
    {
        $v^\star\left(y_{i}\right) \gets \left(\rho/\{y_1,\ldots,y_{i-1}\}\right)(y_i)$\tcp*{greedily increase coordinates}
    }
    \; 
    \Return{$v^\star$}\;
\end{algorithm}

\paragraph{Polymatroids} A \emph{polymatroid rank function} $\rho:2^\univ\to\mathbb{Z}$ is a monotone, normalized, submodular function.
Though integrality is not strictly required, the polymatroid rank functions used in this work will have integral codomains.
A matroid rank function $r_\mathcal{M}$ is simply a \emph{subcardinal} polymatroid rank function, i.e. $r_\mathcal{M}(X) \leq |X|$ always holds.
In this way, polymatroids generalize matroids.
They can also be viewed as a multiset generalization of matroids as well.
These perspectives lead to a few different ways of constructing matroids from polymatroids, which we will see in Section~\ref{subsec:mps-greedoid} and Section~\ref{subsec:scaffold-sub}.
The \emph{base polytope} of a polymatroid is the set of nonnegative vectors in $\mathbb{R}^\univ$ saturating the total rank and satisfying a collection of rank inequalities over all subsets as follows,
\[
    \left\{v \in\mathbb{R}^\univ_{+}\Biggm\vert\lVert v \rVert_1 =\rho(\univ)\text{ and }(\forall X\subseteq \univ)\;\sum_{y \in X} v(y) \leq \rho(X)\right\}.
\]
Famously, when introducing polymatroids Edmonds gave a greedy algorithm for solving linear maximization over the base polytope \cite{edmonds2003submodular}.
This is shown in Algorithm~\ref{alg:edmonds}.
Finally, the \emph{span} operator $\sigma_\rho:2^\univ\to 2^\univ$ is defined as $\sigma_\rho(X) \triangleq\{y \in \univ\mid \rho(X + y) = \rho(X)\}$,
i.e.\ those elements for which the marginal return follows $(\rho/X)(y) = 0$. 
The span is a closure operator, that is it is \emph{monotone} (i.e. $X\subseteq Y$ implies $\sigma_\rho(X) \subseteq \sigma_\rho(Y)$), \emph{idempotent} (i.e. $\sigma_\rho(\sigma_\rho(X)) = \sigma_\rho(X)$, and \emph{extensive} (i.e. $X\subseteq \sigma_\rho(X)$).
We note that the $\rho$-closed sets are closed under intersection, i.e. $\sigma_\rho(X) \cap \sigma_\rho(Y) = \sigma_\rho(X\cap Y)$, and that this is a necessary and sufficient description of any closure operator.
And so, those sets $X \subseteq \univ$ such that $X = \sigma_\rho(X)$ are called $\rho$-\emph{closed}.
Notice, $\sigma_{r_\mathcal{M}} = \sigma_\mathcal{M}$ for a matroid $\mathcal{M}$.
A defining feature of the matroid span is the \emph{Steinitz-Maclane exchange property}, which is $y,z \notin \sigma_\mathcal{M}(X)$ and $y \in \sigma_\mathcal{M}(X+z)$ implies $z \in\sigma_\mathcal{M}(X)$.
This is the formal reason why we may use matroid span operators to encode substitutability in constraint structures, as the Steinitz-Maclane exchange property ensures that $z$ ``capturing'' $y$ in the increased span $\sigma_\mathcal{M}(X+z)$ is a symmetric relationship.
Finally, a \emph{loop} $y \in \univ$ is an element such that $\rho(y) = 0$.
This means that $v(y) = 0$ for every $v$ in the base polytope, by the law of diminishing returns.

\paragraph{Lattices and Order}
Let $\mathscr{P} = (\mathcal{X}, \sqsubset)$ be a poset with reflexive closure `$\sqsubseteq$.'
A \emph{linear extension} is an ordering of $\mathcal{X}$ which agrees with `$\sqsubset$,' that is a word $\mathbf{y}_1\mathbf{y_2}\ldots$ of length $|\mathcal{X}|$ such that $\mathbf{y}_i \sqsubset \mathbf{y}_j$ implies $i < j$.
A set $I \subseteq \mathscr{P}$ is an (order) \emph{ideal} if $\mathbf{z} \in I$ and $\mathbf{y} \sqsubset \mathbf{z}$ implies $\mathbf{y} \in I$.
The ideal generated by an arbitrary set $S$ is,
$$\ideal S \triangleq \{\mathbf{y}\in\mathcal{X}\mid(\exists \mathbf{z} \in S)\; \mathbf{y}\sqsubseteq \mathbf{z}\}.$$
A \emph{principal ideal} $\ideal \mathbf{y}$ is generated by a single poset element.
For $\mathbf{x},\mathbf{y}\in\mathscr{P}$ the \emph{meet} $\mathbf{x} \sqcap \mathbf{y}$ is the binary infimum while the \emph{join} $\mathbf{x} \sqcup \mathbf{y}$ is the binary supremum whenever unique such elements exist.
A poset is a \emph{lattice} when a meet and join is defined for every pair.
The \emph{covering relation}, which we denote by `$\prec$' for any poset, is such that $\mathbf{x} \prec \mathbf{y}$ if and only if $\mathbf{x} \sqsubset \mathbf{y}$ and there is no $\mathbf{z}$ with $\mathbf{x} \sqsubset \mathbf{z} \sqsubset \mathbf{y}$.
One visualizes posets via \emph{Hasse diagrams}, consisting of nodes corresponding to the elements of the ground set $\mathcal{X}$ and line segments such that there is a line directed upwards from $\mathbf{x}$ to $\mathbf{y}$ if and only if $\mathbf{x} \prec \mathbf{y}$.
A poset $\mathscr{P}$ is \emph{graded} whenever there exists a function $g:\mathscr{P}\to\mathbb{Z}_+$ such that $\mathbf{x} \prec \mathbf{y}$ implies $g(\mathbf{x}) = g(\mathbf{y}) - 1$.
A grade taking value zero on bottom elements is the \emph{height function}.
Finally, a lattice is \emph{semimodular} if and only if $\mathbf{x} \succ \mathbf{x} \sqcap \mathbf{y}$ implies $\mathbf{y} \prec \mathbf{x} \sqcup \mathbf{y}$.
We note semimodular lattices are always graded.

\paragraph{Greedoids}
Now we review greedoids.
Our primary sources are \cite{korte2012greedoids,bjorner1992introduction}.
Let $\univ^*$ be all finite sequences over $\univ$. 
Each $\alpha \in \univ^*$ is a \emph{word} $\alpha = y_1\ldots y_k$ composed of \emph{letters} from $\univ$.
A \emph{language} $\glang\subseteq \univ^*$ is a set of words.
We only examine \emph{simple} languages which do not contain words with repeated letters.
To model problems of interest using languages, one can let letters symbolize decisions which can be made over a process, while words give the \emph{feasible} arrangements of decisions.
The \emph{support} $\widetilde{\alpha}$
is the set of letters forming $\alpha$, and we refer to the length of a
word via $|\alpha|$.
Longest words in $\glang$ are called \emph{basic}.
Finally, a \emph{loop} $y \in \univ$ is a letter such that $y \notin\widetilde{\alpha}$ for all $\alpha \in\glang$, i.e. a loop does not appear in any word.
A \emph{greedoid} can be defined as a language generalization of a matroid.
\begin{definition}[Greedoid]\label{def:greedoid}
    A \emph{greedoid} $\glang$ over $\univ$ is a nonempty simple language satisfying,
    \begin{enumerate}
        \item \emph{(Hereditary) }$\alpha\beta \in \Lambda$ implies the prefix $\alpha$ is feasible, i.e. $\alpha\in\Lambda$,
        \item \emph{(Exchange) }For all $\alpha, \beta \in \Lambda$, if $|\alpha| > |\beta|$ then there exists $x\in\widetilde{\alpha}$ such that $\beta x \in \Lambda$.
    \end{enumerate}
\end{definition}

We give an example relevant to our MPS construction in Section~\ref{sec:mps}, and point the reader towards \cite{korte2012greedoids,bjorner1992introduction} for more examples.
Non-emptiness and the hereditary axiom together imply $\epsilon \in \glang$, and exchange makes all basic words equal in length.
A prominent class is \emph{interval greedoids}, which possess a stronger variant of exchange:
call a simple $\alpha'\in\univ^*$ a \emph{subword} of $\alpha\in\glang$ if and only if
$\widetilde{\alpha}' \subseteq \widetilde{\alpha}$ and the letters of $\alpha'$ occur in an order agreeing with $\alpha$ (i.e. $x$ appears before $y$ in $\alpha'$ implies $x$ appears before $y$ in $\alpha$).
Interval greedoids guarantee feasible words can always be augmented by subword of a longer feasible word to achieve length at least that of the longer word.
\begin{definition}[Interval Property]
    Let $\glang$ be a greedoid.
    Then, $\glang$ has the \emph{interval property} if for all $\alpha ,\beta \in \glang$ with $|\beta| > |\alpha|$, there exists a subword $\beta'$ of $\beta$ of length $|\beta'| \geq |\beta| - |\alpha|$ with $\alpha \beta'\in\glang$.
\end{definition}

The contraction minor $\glang/\alpha$, where $\alpha \in\glang$, is the language of simple words which can be concatenated onto $\alpha$ while maintaining feasibility in $\glang$, i.e. $\glang/\alpha \triangleq \{\beta \in \univ^*\mid \alpha\beta \in \glang\}.$
It can be verified $\glang/\alpha$ is a greedoid, and we call the single length words $y \in \glang/\alpha$ \emph{continuations} since they are exactly $y \in \univ$ with $\alpha y \in \glang$.
Define the equivalence relation $\alpha \sim \beta$ by $\glang/\alpha = \glang/\beta$.
A \emph{flat} is an equivalence class of $\glang/\mathord{\sim}$, and the class associated with $\alpha$ is $[\alpha]$.
One can order the flats by a kind of reachability:
let $[\alpha] \sqsubseteq [\beta]$ if and only if there exists $\gamma\in\glang/\alpha$ where $\alpha \gamma\in\glang$ and $\alpha \gamma \sim \beta$.
For interval greedoids, the resulting poset $(\glang/\mathord{\sim},\sqsubseteq)$ is a semimodular lattice.
Our MPS is a kind of interval greedoid, and a Hasse diagram corresponding to a lattice of flats is given in Figure~\ref{fig:spg-example}.
Finally, we define the \emph{flat support} as the union of supports of words in a flat.
This is given by $\kappa:\glang/\mathord{\sim}\to 2^\univ$ defined as $\kappa[\alpha] \triangleq \bigcup_{\beta \sim\alpha}\widetilde{\beta}$.
We choose this symbol because a flat support arises from the greedoid's \emph{kernel closure operator} (see Ch. 5 of \cite{korte2012greedoids}).
A special property of interval greedoids is that $\kappa$ is invertible, so ordering the flat supports by containment gives a lattice isomorphic to $(\glang/\mathord{\sim},\sqsubseteq)$.

A \emph{polymatroid greedoid}, introduced in \cite{korte1985polymatroid}, is such that there exists a polymatroid rank function $\rho:2^\univ\to\mathbb{Z}$ satisfying,
$$\glang = \{x_1\ldots x_k\in\univ^*\mid (\forall i \in \{1,\ldots,k\})\;\rho(\{x_1,\ldots,x_i\}) = i\}.$$
Examples include matroids, poset antimatroids, and certain ordered geometries; see Ch. 7 of \cite{korte2012greedoids}.
Following \cite{streit2024polymatroid}, in this context we call $\rho$ a \emph{representation} as it encodes how one builds feasible words in $\glang$.
This is, for all $\alpha \in\glang$ and $y \notin \widetilde{\alpha}$, $\alpha y \in\glang$ if and only if $(\rho/\widetilde{\alpha})(y) = 1$.
But more importantly, the representation defines a global submodular structure over the greedoid, which applies to both those sets corresponding to and those which do not correspond to supports of feasible words.
This submodular structure is a property given by matroid rank functions, for example, which is lost in generalizing to greedoids.
The concept of representation is notable because it restores this feature.
In \cite{streit2024polymatroid}, the authors define a subclass called \emph{strong polymatroid greedoids} by those guaranteeing the function $\gtest$, where $\gtest(X)$ is the height of a minimal flat support containing $X$ in the lattice $(\kappa(\glang/\mathord{\sim}),\subseteq)$, i.e.,
\[
    \gtest(X) \triangleq \min\{|\alpha|\mid X\subseteq \kappa[\alpha]\},
\]
is a polymatroid rank function representing the greedoid.
When true, $\gtest$ is called the \emph{greatest representation} because it can be shown be coordinatewise larger than all other representations. 
The authors give a combinatorial description of such greedoids: a greedoid is \emph{optimistic} whenever for every basic word $y_1\ldots y_R$ and nonloop $z \in \univ$ there exists an index $i$ making $y_1\ldots y_{i-1} z \in \glang$.
We remark that all polymatroid greedoids are optimistic \cite{streit2024polymatroid,korte1985polymatroid}.
Then, $\glang$ is a strong polymatroid greedoid if and only if it is an optimistic interval greedoid whose flat supports are closed under intersection \cite{streit2024polymatroid}, i.e. $\kappa(F) \cap \kappa(F') = \kappa(F\sqcap F')$ for all $F,F'\in\glang/\mathord{\sim}$.
This shows the existence of a closure operator $\cl:2^\univ\to 2^\univ$,
\[
    \cl(X) \triangleq \bigcap\{\kappa[\alpha]\mid X\subseteq \kappa[\alpha]\},
\]
which preserves the greedoid structure, as the $\kappa$-image and $\cl$-images are equal (though this does not imply $\kappa = \cl$).
This is a closure operator, because the definition of $\cl$ implies its image is closed under intersection.
And, $\cl$ also simplifies the definition of $\gtest$ to make $\gtest(X)$ the height of $\cl(X)$ in $(\kappa(\glang/\mathord{\sim}),\subseteq)$.
By inspection then $\sigma_{\gtest} = \cl$, and so the greatest representation gives us functional access to this closure system.
For this reason, $\gtest$ and $\cl$ are meaningful tools for algorithm design.
In Section~\ref{subsec:mps-greedoid}, we show every MPS is a strong polymatroid greedoid (and vice versa) so that our approximation algorithms and analysis can appeal to this theory.

\section{Matroidal Prerequisite Systems}\label{sec:mps}

Now we formally define the matroidal prerequisite system.
This is an interaction between a poset and a matroid, with the poset encoding prerequisites and the matroid encoding a substitutability structure, combined with a certain compatibility condition ensuring this interaction is well-behaved.

\begin{definition}[Matroidal Prerequisite System (MPS)]
    Fix a finite alphabet $\univ$, and let $(\univ,\mathcal{M})$ and $(\univ,\leadsto)$ be a (loopless) matroid and poset respectively.
    Define $\mathcal{L}(\mathcal{M},\leadsto)$ as the set of simple words $x_1\ldots x_\ell \in \univ^*$ with $\{x_1,\ldots,x_\ell\}\in\mathcal{M}$ and,
    \begin{equation}\label{eq:feasibility}
        y \leadsto x_i \implies y \in \sigma_\mathcal{M}(\{x_1,\ldots,x_{i-1}\})\text{ or }\sigma_\mathcal{M}(\{x_1,\ldots,x_{i-1},y\}) = \sigma_\mathcal{M}(\{x_1,\ldots,x_{i-1},x_i\}),
    \end{equation}
    holding for all $i \in \{1,\ldots,\ell\}$.
    Then, $\mathcal{L}(\mathcal{M},\leadsto)$ is a \emph{matroidal prerequisite system} if and only if,
    \begin{equation}\label{eq:compatibility}
        \alpha \in \mathcal{L}(\mathcal{M},\leadsto)\text{ and }z \in \sigma_\mathcal{M}(\widetilde{\alpha}) \implies (\forall y \leadsto z)\;y \in \sigma_\mathcal{M}(\widetilde{\alpha}).
    \end{equation}
\end{definition}

This is, $\mathcal{L}(\mathcal{M},\leadsto)$ is words with independent supports such that every prerequisite of a letter $y$ in a feasible word is either spanned by prefix ending at $y$ (exclusive) in the matroid or is substitutable with $y$ against that prefix.
Call the words of $\mathcal{L}(\mathcal{M},\leadsto)$ \emph{feasible}.
Then, the compatibility condition in Eq.~\ref{eq:compatibility} states that whenever a letter is covered by the matroid span of a feasible word, so must all of its prerequisites.
This has good semantic meaning with the problems we are interested in modeling (i.e. a decision which has been made redundant is such that all of its prerequisite decisions are redundant as well), and is required to verify $\mathcal{L}(\mathcal{M},\leadsto)$ is an optimistic greedoid (which we need to use $\gtest$, $\cl$, and other concepts inducing correct behavior in our algorithms).
We also implicitly assume the matroid is loopless so as to prevent complications or cumbersome notation caused by loops.
This can be satisfied by simply removing letters with matroid rank zero from the groundset.
Finally, as discussed in Section~\ref{sec:intro}, this abstracts matroids and poset antimatroids by considering the MPS defined over an empty order and free matroid, respectively.

Before continuing, we list a few properties.
Firstly, every basic word of $\mathcal{L}(\mathcal{M},\leadsto)$ is always supported on a basic set (largest independent set) in $\mathcal{M}$.
This is because, for all $\alpha \in \mathcal{L}(\mathcal{M},\leadsto)$, any minimal (w.r.t. `$\leadsto$') letter $y \in \univ\setminus\sigma_\mathcal{M}(\widetilde{\alpha})$ is such that minimality guarantees all its prerequisites are covered by $\sigma_\mathcal{M}(\widetilde{\alpha})$.
Moreover, $y \notin \sigma_\mathcal{M}(\widetilde{\alpha})$ also implies $\widetilde{\alpha}+ y$ is independent, and so $\alpha y \in\mathcal{L}(\mathcal{M},\leadsto)$.
Such a $y$ always exists whenever $\sigma_\mathcal{M}(\widetilde{\alpha})\neq \univ$, which occurs exactly when $\widetilde{\alpha}$ is not a basic set of the matroid.
Secondly, this implies every MPS defined by a loopless matroid is also loopless.
Specifically, since $\sigma_\mathcal{M}(\widetilde{\alpha}) = \univ$ for every basic word $\alpha$, for all $y\in\univ$ there exists an $\alpha$-prefix $\alpha'z$ such that $y \in \sigma_\mathcal{M}(\widetilde{\alpha}' + z)\setminus\sigma_\mathcal{M}(\widetilde{\alpha}')$.
Then, since $y$ is $\widetilde{\alpha}'$-substitutable with all its prerequisites by the compatibility condition (Eq.~\ref{eq:compatibility}), it follows that $\alpha'y \in \mathcal{L}(\mathcal{M},\leadsto)$.
In what comes, we show an MPS is optimistic, and so this looplessness will mean that given any basic word and letter $y\in\univ$ we can assert the existence of a prefix possessing $y$ as a continuation.

As a technical example, let $\mathcal{U}^R_n$ be a \emph{uniform matroid} of rank $R$.
That is, every set of cardinality at most $R$ is independent.
$\mathcal{L}(\mathcal{U}^R_n,\leadsto)$ always satisfies the compatibility condition, regardless of the prerequisite order $(\univ,\leadsto)$.
This is because the only independent sets with span not equal to themselves are basic.
So, basic words in $\mathcal{L}(\mathcal{M},\leadsto)$ are the only feasible words whose corresponding matroid span contains any letters not from their support.
Therefore, for all $\alpha \in \mathcal{L}(\mathcal{M},\leadsto)$ which are not basic we have $z \in \sigma_\mathcal{M}(\widetilde{\alpha})$ implies $z \in \widetilde{\alpha}$, and so all prerequisites $y \leadsto z$ are contained in $\sigma_\mathcal{M}(\widetilde{\alpha})$ by Eq.~\ref{eq:feasibility}, while if $\alpha$ is basic then it is supported on a basic set of $\mathcal{U}^R_n$ and $\sigma_\mathcal{M}(\widetilde{\alpha}) = \univ$ trivially makes the conclusion of Eq.~\ref{eq:compatibility} true.
This makes $\mathcal{L}(\mathcal{U}^R_n,\leadsto)$ a matroidal prerequisite system.
We can use this to model cardinality constrained optimization over a poset antimatroid.
This is equivalent to cardinality constrained maximum closure over a directed acyclic graph \cite{picard1976maximal} and unit-weight poset constrained knapsacks \cite{johnson1983knapsacks,samphaiboon2000heuristic,kolliopoulos2007partially}, so an MPS abstracts these problems as well.

\begin{example}[Cardinality Constrained Optimization on a Poset Antimatroid]\label{ex:cc-pam}
    Let $(\univ,\leadsto)$ be a poset with $|\univ| = n$, $w:\univ\to\mathbb{R}$ a nonnegative weight function, and $k\leq n$ an integer.
    Our task is to find an ideal of size $k$ of maximum cumulative weight, or equivalently finding a $k$-length prefix of a linear extension with maximum cumulative weight.
    Add an additional ``sentinel letter'' $t$ by letting $\univ' = \univ + t$, and consider the matroid $(\univ',\mathcal{U}^{k+1}_{n+1})$.
    Extend the order to $(\univ',\leadsto')$ by letting,
    \begin{align*}
        y \leadsto z \implies y \leadsto' z,&&\text{and,}&& (\forall y \in \univ)\;y\leadsto' t.
    \end{align*}
    This is, the derived order inherits the relations given by `$\leadsto$' while making every letter lesser than the sentinel.
    Finally, make the derived weight function $w':\univ'\to\mathbb{R}$ by letting $w'(y) = w(y)$ for all $y \in \univ$, and $w'(t) = \max_{y \in \univ}w(y) + 1$.
    It can be seen that any word of maximum weight over the MPS $\mathcal{L}(\mathcal{U}^{k+1}_{n+1},\leadsto')$ is a basic word ending in $t$. It follows that the first $k$ letters of such a word will be an ideal of $(\univ,\leadsto)$ of size $k$ maximizing $X\mapsto\sum_{y \in X}w(y)$.
    And so, this solves our original problem.
\end{example}

\begin{remark}\label{rem:dis-super}
    The introduction of the sentinel letter was to mitigate a certain quirk, being words $\alpha \in \mathcal{L}(\mathcal{M},\leadsto)$ of length one less than $r_\mathcal{M}(\univ)$ are such that every letter $y\notin \sigma_\mathcal{M}(\widetilde{\alpha})$ is substitutable with any of its prerequisites not already spanned by $\widetilde{\alpha}$, since $\sigma_\mathcal{M}(\widetilde{\alpha}+y) = \univ$ will hold.
This means basic words of $\mathcal{L}(\mathcal{U}^{k+1}_{n+1},\leadsto')$, are not supported on ideals of $(\univ',\leadsto')$ even though the nonbasic words are supported on ideals of $(\univ,\leadsto)$.
This gives an easy way to see that matroidal prerequisite systems are not simply the intersection of a poset antimatroid and a matroid, since if this were the case then the basic words of $\mathcal{L}(\mathcal{U}^{k+1}_{n+1},\leadsto')$ would be prefixes of linear extensions.
So, our MPS is a distinct structure from distributive supermatroids by the description in Ch. 9 of \cite{korte2012greedoids}.
\end{remark}

\begin{figure}
    \centering
    \begin{subfigure}{.48\textwidth}
        \centering

\begin{tikzpicture}[scale=0.92, every node/.style={font=\small}]
  \draw[thick] (0,0) -- (8,0);
  \draw[thick] (0,0) .. controls (2,2.25) and (6,2.25) .. (8,0); 

  \foreach \x/\name/\coord in {
    1/$s_1$/Sone,
    2.2/$a_1$/Aone,
    3.4/$s_2$/Stwo,
    4.6/$s_3$/Sthree,
    5.8/$a_2$/Atwo,
    7/$s_4$/Sfour
  }{
    \coordinate (\coord) at (\x,0);
    \filldraw[black] (\coord) circle (1.6pt);
    \node[above=2pt] at (\coord) {\name};
  }

  \draw[red,thick] (2.2,0) circle (3.3pt);
  \draw[red,thick] (5.8,0) circle (3.3pt);

  \filldraw[black] (6,1.3255) circle (1.6pt);
  \node[above=2pt] at (6,1.3255) {$t_2$};
  \draw[blue,thick] (6,1.3255) circle (3.3pt);
  
  \filldraw[black] (4,1.6825) circle (1.6pt);
  \node[above=2pt] at (4,1.6825) {$t_1$};
  \draw[blue,thick] (4,1.6825) circle (3.3pt);

  \node[below=11pt] (VecSone)   at (1,.25)
    {$\left[\begin{smallmatrix}1\\0\\0\end{smallmatrix}\right]$};

  \node[below=11pt] (VecAone)   at (2.2,0.25)
    {$\left[\begin{smallmatrix}1\\0\\1\end{smallmatrix}\right]$};

  \node[below=11pt] (VecStwo)   at (3.4,0.25)
    {$\left[\begin{smallmatrix}0\\1\\0\end{smallmatrix}\right]$};

  \node[below=11pt] (VecSthree) at (4.6,0.25)
    {$\left[\begin{smallmatrix}0\\0\\1\end{smallmatrix}\right]$};

  \node[below=11pt] (VecAtwo)   at (5.8,0.25)
    {$\left[\begin{smallmatrix}0\\1\\1\end{smallmatrix}\right]$};

  \node[below=11pt] (VecSfour)  at (7,0.25)
    {$\left[\begin{smallmatrix}1\\1\\0\end{smallmatrix}\right]$};

  \def\LabelOffset{0.55}

  \coordinate (Ltop1) at (1,\LabelOffset);
  \coordinate (Ltop2) at (2.2,\LabelOffset);
  \coordinate (Ltop3) at (3.4,\LabelOffset);

  \coordinate (Rtop1) at (4.6,\LabelOffset);
  \coordinate (Rtop2) at (5.8,\LabelOffset);
  \coordinate (Rtop3) at (7,\LabelOffset);

  \begin{pgfonlayer}{bg}
    \node[
      draw,
      dashed,
      rounded corners=5pt,
      very thin,
      inner sep=0pt,
      fit=(Ltop1)(Ltop2)(Ltop3)(VecSone)(VecAone)(VecStwo)
    ] (LeftBox) {};

    \node[
      draw,
      dashed,
      rounded corners=5pt,
      very thin,
      inner sep=0pt,
      fit=(Rtop1)(Rtop2)(Rtop3)(VecSthree)(VecAtwo)(VecSfour)
    ] (RightBox) {};

    \node[font=\scriptsize, anchor=north] at (LeftBox.south)  {West Block};
    \node[font=\scriptsize, anchor=north] at (RightBox.south) {East Block};
  \end{pgfonlayer}
\end{tikzpicture}
        \caption{A bridge with sensor placements. Modal signatures are underneath each sensor. Baseline sensors are outlined in red, sync hardware is outlined in blue, and the remaining sensors are for anomaly detection.}
    \end{subfigure}
    \hfill
    \begin{subfigure}{.48\textwidth}
        \centering
        \scalebox{.7}{\begin{tikzpicture}[
    edge/.style={draw}
]
\tikzstyle{S}=[rectangle, draw=black, rounded corners=5pt]

\matrix[matrix of math nodes,
  left delimiter={[},
  right delimiter={]},
  nodes in empty cells,
  column sep=0.1em,
  row sep=0.1em,
  ampersand replacement=\&] (m) {
      1 \& 1 \& 0 \& 0 \& 0 \& 1 \\
      0 \& 0 \& 1 \& 0 \& 1 \& 1 \\
      0 \& 1 \& 0 \& 1 \& 1 \& 0 \\
};

\node[above=4pt of m-1-1] (lab1) {$s_{1}$};
\node[above=4pt of m-1-2] (lab2) {$a_{1}$};
\node[above=4pt of m-1-3] (lab3) {$s_{2}$};
\node[above=4pt of m-1-4] (lab4) {$s_{3}$};
\node[above=4pt of m-1-5] (lab5) {$a_{2}$};
\node[above=4pt of m-1-6] (lab6) {$s_{4}$};

\node[left=6pt of m] (Mnode) {$M$};
\node[right=6pt of m] (plusnode) {$\oplus\, \mathcal{U}^1(\{t_1,t_2\})$};

\draw[decorate,decoration={brace,mirror,amplitude=10pt},thick]
  ([yshift=-8pt,xshift=-20pt]m.south west) -- ([yshift=-8pt,xshift=75pt]m.south east)
  node[midway,yshift=-20pt] {$(\Sigma,\mathcal{M})$};


\node[S,left=85pt of m,yshift=-50pt] (t1) {$t_1$};
\node[S] (t2) [right=of t1,xshift=-2mm] {$t_2$};
        
\node[S] (a1) [above=of t1] {$a_1$};
\node[S] (a2) [above=of t2] {$a_2$};
        
\node[S] (s1) [above left=of a1,xshift=5mm] {$s_1$};
\node[S] (s2) [above right=of a1,xshift=-5mm] {$s_2$};
\node[S] (s3) [above left=of a2,xshift=5mm] {$s_3$};
\node[S] (s4) [above right=of a2,xshift=-5mm] {$s_4$};

\draw[decorate,decoration={brace,amplitude=10pt},thick]
  ([yshift=10pt]s1.north west) -- ([yshift=10pt]s4.north east)
  node[midway,yshift=20pt] {$(\Sigma,\leadsto)$};
        
\draw[edge] (t1) -- (a1);
\draw[edge] (t2) -- (a1);
\draw[edge] (t1) -- (a2);
\draw[edge] (t2) -- (a2);
        
\draw[edge] (a1) -- (s1);
\draw[edge] (a1) -- (s2);
\draw[edge] (a2) -- (s3);
\draw[edge] (a2) -- (s4);

\end{tikzpicture}}
        \caption{A MPS following Example~\ref{ex:ad} is given by a poset like above and the direct sum of a linear matroid on $\mathrm{GF}(2)$ over the modal signatures with the rank 1 uniform matroid over $\{t_1,t_2\}$.}
    \end{subfigure}
    \caption{Left shows a bridge. Following Example~\ref{ex:ad}, the modal signatures of the baseline sensor of each block must be spanned before an anomaly detection sensor is placed in that block. Furthermore, we must also place sync hardware before any other sensor. This leads to the poset structure on the right. For budget reasons, we can only place one sync hardware. A linear matroid encodes the substitution structure of the modal signatures as well as a linear independence condition in the measurements (which can help with sensor fusion \cite{cohen2014sensor}), so the matroid is given by a direct sum like on the right. To save space, we verify the compatibility condition in Appendix~\ref{app:ver}. So, we have an MPS. Some examples of basic words are $t_1a_1 s_1 s_2$, $t_1a_2s_4s_1$, and $t_2a_1a_2s_1$.}\label{fig:bridge}
\end{figure}
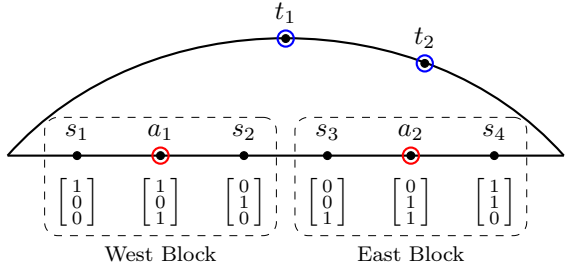
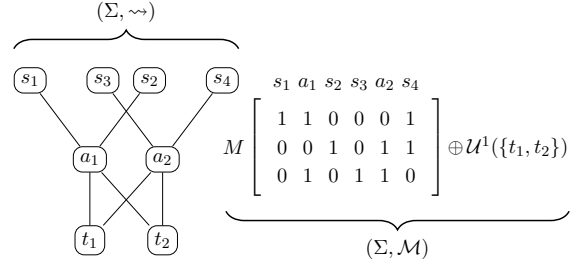

We give a few less technical examples showing the role a MPS can play in applied problems like sensor placement, experiment design, and project scheduling.
We use the first example, anomaly detection via sensor networks, to motivate a small concrete instance of an MPS described in Figure~\ref{fig:bridge}.

\begin{example}[Sensor Networks for Anomaly Detection]\label{ex:ad}
    Sensor placement to maximize mutual information \cite{krause2008near} is a classic instance of monotone submodular maximization, with particularly natural matroid constraints encoding budget, network interference, etc. (see \cite{powers2017constrained,cohen2014sensor} for example).
    One specific task is using sensors to detect anomalies; for example the presence of errant vibration frequencies indicates damage in a structure \cite{bigoni2020systematic}.
    However, it is not possible to interpret anomalies without a frame of reference given by baseline measurements.
    This creates a prerequisite structure, where modes corresponding to a baseline must be spanned before anomaly detection can occur.
    Other prerequisite constraints which are operational may coexist, such as requiring clock synchronization hardware is placed before all other hardware to ensure a correct boot procedure. 
\end{example}

An instance of detecting anomalous frequencies on a bridge is shown in Figure~\ref{fig:bridge}.
A modal signature is a binary vector encoding the capabilities of each sensor, and a linear matroid in combination with a poset encodes the requirement that baseline sensor modal signatures in each block must be spanned before anomaly detecting sensors in those blocks can be placed, along with the requirement that clock synchronization hardware is placed before any sensor.
An interesting feasible word is $t_1a_2s_4s_1$.
Observe, even though the baseline sensor $a_1$ was never placed in the East block, we can still place $s_1$ after $t_1a_2s_4$ since the modal signature of $a_1$ is in the span of $\{a_2,s_4\}$.
This models the functionality contributed by $a_1$ becoming redundant in the presence of $a_2$ and $s_4$.

\begin{remark}
    Let $w:\univ\to\mathbb{R}$ be a weight function.
    In the context of greedoids (which we will show an MPS is), the greedy algorithm consists of maintaining a word $\alpha$, initially $\alpha =\epsilon$, and selecting the continuation $y$ of maximum weight for the update $\alpha \gets \alpha y$.
    Observe, with $w(t_2) = w(a_2) = 0$, $w(t_1) =w(s_1) = w(s_2)= w(a_1) = 1$, and $w(s_3) = w(s_4) = 3$, greedy makes $t_1 a_1 s_1 s_3$.
    The value of this solution is $w(t_1) + w(a_1) + w(s_1) + w(s_3) = 6$.
    Another feasible word is $t_1 a_2 s_3 s_4$, with value $w(t_1) + w(a_2) + w(s_3) + w(s_4) = 8$.
    So, the greedy algorithm is not guaranteed to succeed.
    Generally, greedy can only guarantee an approximation about equal to the matroid rank; see Appendix~\ref{app:greedy-fail}.
\end{remark}

Now, we describe the application of an MPS to experiment design over a conditional parameter space.
In this example, the dependency hierarchies of parameters of statistical models can induce a more complex poset structure than the previous example.

\begin{example}[Experiment Design over Conditional Parameter Spaces]
    D-optimal design, i.e. choosing experiments to maximize Fisher information about model parameters, can also be modeled as monotone submodular maximization \cite{bouhtou2010submodularity}.
    Matroid constraints may appear \cite{brown2022determinant}, for example a uniform matroid encodes a budget constraint.
    However, in many settings an experiment's parameters will be conditioned on other parameters, thereby forming a complex hierarchy.
    For example, the number of weights in a neural network will vary in the number of layers.
    Since experiments on the full parameter space may be expensive, one can instead probe coarse approximations which inform and iteratively refine towards more costly finer probes (see \cite{sen2018multi} for example).
    A poset can maintain this invariant, either for operational or efficiency concerns arising from cost.
\end{example}

Finally, where the last two examples showcase the added value of partial order structure, our final example takes a classically poset constrained problem an enriches through the interaction of matroids.
This is task scheduling for a project, which  would be a typical motivation of the sort of structure we described in Example~\ref{ex:cc-pam} (see \cite{kolliopoulos2007partially} for example).
Generalizing with an MPS creates the opportunity to model more complex substitution structures.

\begin{example}[Project Scheduling]
    An engineering project has tasks with dependencies; such dependencies form a natural partial order structure.
    The goal is to select a bounded number of tasks to complete such that all dependency relations and stakeholder value is maximized, which can be encoded via an additive or submodular objective.
    However, in a similar spirit to \cite{mohring2004scheduling}, we recognize that one task depends on another not because of the identity of the task, but because of some functionality provided by the task.
    Such functionality may be substituted by an alternative, or covered within the interaction of many other tasks.
    For that reason, one could use a matroid to encode this substitutes structure, and an MPS to reasonably encode the interaction between the dependency structure of a partial order and the substitute structure of the matroid.
\end{example}

In Section~\ref{subsec:mps-greedoid}, we examine the greedoid theory of an MPS.
In particular, we not only show that an MPS is always a greedoid, but that every MPS can be described as a strong polymatroid greedoid and vice versa.
The utility of this is best seen in Section~\ref{subsec:fill}, where we use the greatest representation and closure operator of strong polymatroid greedoids to give procedures for ``filling-in'' and ``projecting'' partial optimization solutions on to the basic words of an MPS.
This is Lemma~\ref{lem:projection}, and will be the foundation of our algorithmic techniques in Sections~\ref{sec:scaffold} and~\ref{sec:mag}.
Finally, since we use the greatest representation in our algorithms, we verify that it can be computed efficiently when given a matroid rank oracle in Section~\ref{subsec:mps-gr}.
This is because, even though our algorithms are described in terms of strong polymatroid greedoid concepts, our previous examples show that the MPS is a widely applicable and interpretable constraint model.
And so, while strong polymatroid greedoids give the mechanisms for our algorithms, an MPS grounds our motivations for designing them in the first place.
Therefore, we must verify that given an MPS it is possible to efficiently compute strong polymatroid greedoid concepts under reasonable assumptions.

\subsection{Greedoid Theoretic Aspects}\label{subsec:mps-greedoid}

We now show an MPS is not only a greedoid, but that a greedoid is a strong polymatroid greedoid if and only if there exists a matroid and poset giving its feasible words as an MPS.
Showing an MPS is a strong polymatroid greedoid will use the description as optimistic interval greedoids whose flat supports are closed under intersection \cite{streit2024polymatroid}.
We first verify optimism and the interval property.

\begin{proposition}\label{prop:opt+int}
    Let $\mathcal{L}(\mathcal{M},\leadsto)$ be a matroidal prerequisite system.
    Then, $\mathcal{L}(\mathcal{M},\leadsto)$ is an optimistic interval greedoid.
\end{proposition}
\begin{proof}
    From the definition $\mathcal{L}(\mathcal{M},\leadsto)$ is nonempty and hereditary.
    Let $\alpha,\beta \in \mathcal{L}(\mathcal{M},\leadsto)$ with $|\beta| > |\alpha|$.
    Select the first $z_1$ appearing in $\widetilde{\beta}$ with $z_1\notin{\sigma_\mathcal{M}}(\widetilde{\alpha})$.
    We show $\alpha z_1$ is feasible.
    Well, $z_1\notin {\sigma_\mathcal{M}}(\widetilde{\alpha})$ implies $\widetilde{\alpha} + z_1$ is independent.
    And, if $\beta'z_1$ is the $\beta$-prefix ending at $z_1$, we see $\widetilde{\beta}'\subseteq{\sigma_\mathcal{M}}(\widetilde{\alpha})$ by construction, thus $\sigma_\mathcal{M}(\widetilde{\beta}')\subseteq\sigma_\mathcal{M}(\widetilde{\alpha})$.
    So, for every $y\leadsto z_1$ such that $y \notin{\sigma_\mathcal{M}}(\widetilde{\alpha})$, it follows that $y \notin {\sigma_\mathcal{M}}(\widetilde{\beta}')$.
    Since $\beta' \in \mathcal{L}(\mathcal{M},\leadsto)$, as we have already shown the hereditary property, it must be that $y$ and $z_1$ are $\widetilde{\beta}'$-substitutable.
    Hence, 
    \[
        {\sigma_\mathcal{M}}(\widetilde{\alpha}+y) = {\sigma_\mathcal{M}}(\widetilde{\alpha} \cup {\sigma_\mathcal{M}}(\widetilde{\beta}' + y)) = {\sigma_\mathcal{M}}(\widetilde{\alpha} \cup {\sigma_\mathcal{M}}(\widetilde{\beta}' + z_1) = {\sigma_\mathcal{M}}(\widetilde{\alpha} + z_1),
    \]
    by ${\sigma_\mathcal{M}}(\widetilde{\beta}')\subseteq{\sigma_\mathcal{M}}(\widetilde{\alpha})$, and so $y$ and $z_1$ are $\widetilde{\alpha}$-substitutable.
    Thus $\alpha z_1 \in \mathcal{L}(\mathcal{M},\leadsto)$.
    If $|\beta| - |\alpha| > 1$, one can repeat this argument on the first letter $z_2$ in $\beta$ with $z_2\notin{\sigma_\mathcal{M}}(\widetilde{\alpha z}_1)$ to show $\alpha z_1 z_2 \in \mathcal{L}(\mathcal{M},\leadsto)$ as well.
    Repeating in this way $|\beta| - |\alpha| - 2$ more times shows the existence of the $k = |\beta| - |\alpha|$ length $\beta$-subword $z_1\ldots z_k$ such that $\alpha z_1\ldots z_k \in \mathcal{L}(\mathcal{M},\leadsto)$.
    And so, $\mathcal{L}(\mathcal{M},\leadsto)$ is an interval greedoid.

    Now, let $x_1\ldots x_R \in \mathcal{L}(\mathcal{M},\leadsto)$ be basic.
    Select $y\in\univ$ and let $i$ be the smallest index with $y \in {\sigma_\mathcal{M}}(\{x_1,\ldots,x_i\})$.
    By Steinitz-Maclane exchange property, $x_i$ and $y$ are $\{x_1,\ldots,x_{i-1}\}$-substitutes.
    Well, for every $z \leadsto y$ with $z \notin{\sigma_\mathcal{M}}(\{x_1,\ldots,x_{i-1}\})$,
    the compatibility condition and the observation that $x_i$ and $y$ are $\{x_1,\ldots,x_{i-1}\}$-substitutes combine for, 
    \[
        y \in {\sigma_\mathcal{M}}(\{x_1,\ldots,x_{i-1},x_i\}) \implies z \in {\sigma_\mathcal{M}}(\{x_1,\ldots,x_{i-1},y\}).
    \]
    So, $x_1\ldots x_{i-1}y \in \mathcal{L}(\mathcal{M},\leadsto)$, and optimism follows.
\end{proof}

Now, we make a few observations about the interaction between the flat supports, the matroid, and the prerequisite order.
Specifically, the next lemma first verifies flat supports are always $r_\mathcal{M}$-closed and ideals of $(\univ,\leadsto)$.
These facts then show the flat supports are closed under intersection.

\begin{lemma}\label{lem:int}
    Let $\mathcal{L}(\mathcal{M},\leadsto)$ be a matroidal prerequisite system.
    The following are all true:
    \begin{enumerate}
        \item For all $\alpha \in \mathcal{L}(\mathcal{M},\leadsto)$, $\sigma_\mathcal{M}(\widetilde{\alpha}) = \kappa[\alpha]$,
        \item Every flat support is an ideal in the prerequisite order $(\univ,\leadsto)$,
        \item The flat supports are closed under intersection.
    \end{enumerate}
\end{lemma}
To prove the first item, we need the following statement which follows from the semimodularity of the flats of interval greedoids.
See \cite{bjorner1992introduction} for proof.
\begin{lemma}[\cite{bjorner1992introduction}]\label{lem:bjorner}
    Let $\glang$ be an interval greedoid, and select nonempty words $\alpha x\in \glang$ and $\alpha y \in \glang$.
    If $\alpha x \not\sim \alpha y$, then $\alpha x y\in \glang$ and $\alpha y x\in\glang$.
\end{lemma}
\begin{proof}[Proof of Lemma~\ref{lem:int}]
    For what follows, recall Proposition~\ref{prop:opt+int} shows $\mathcal{L}(\mathcal{M},\leadsto)$ is an optimistic interval greedoid.
    Let $\alpha \in \mathcal{L}(\mathcal{M},\leadsto)$ and select $y \in {\sigma_\mathcal{M}}(\widetilde{\alpha})$.
    Since supports of feasible words are independent, there exists no word in the contraction $\mathcal{L}(\mathcal{M},\leadsto)/\alpha$ containing $y$ in its support.
    So by optimism (since every feasible word can be extended into a basic word), there exists an $\alpha$-prefix $\alpha'z$ with $y \in \mathcal{L}(\mathcal{M},\leadsto)/\alpha'$.
    Let $\alpha'z$ be a longest such prefix, so that $\alpha'y,\alpha'z\in\mathcal{L}(\mathcal{M},\leadsto)$ while neither $\alpha'yz$ or $\alpha'zy$ are feasible.
    By contrapositive of Lemma~\ref{lem:bjorner}, it follows that $\alpha'y \sim \alpha'z$.
    Because our construction makes $[\alpha'z]\sqsubseteq [\alpha]$, it follows that $[\alpha'y]\sqsubseteq[\alpha]$.
    So, there exists $\beta\in\univ^*$ making $\alpha'y\beta \in [\alpha]$, thus $y\in\kappa[\alpha]$.
    This shows ${\sigma_\mathcal{M}}(\widetilde{\alpha})\subseteq \kappa[\alpha]$.
    Now, select $z' \notin {\sigma_\mathcal{M}}(\widetilde{\alpha})$.
    Then, there exists $\beta \in \mathcal{L}(\mathcal{M},\leadsto)/\alpha$ with $z' \in \widetilde{\beta}$:
    There must exist a minimal $y' \in \ideal z' \setminus {\sigma_\mathcal{M}}(\widetilde{\alpha})$ w.r.t. `$\leadsto$,' and this minimal $y'$ must be a continuation of $\alpha$ since its prerequisites are covered by ${\sigma_\mathcal{M}}(\widetilde{\alpha})$.
    Therefore, if $\alpha z' \notin\mathcal{L}(\mathcal{M},\leadsto)$, we can append $y'$ onto $\alpha$.
    If we continue this process, eventually $z'$ will be substitutable with every prerequisite that hasn't been covered, at which point $z'$ will be a continuation.
    This makes $z' \notin \kappa[\alpha]$.
    So, as we've shown ${\sigma_\mathcal{M}}[\alpha] \subseteq \kappa[\alpha]$ and $\univ\setminus{\sigma_\mathcal{M}}(\widetilde{\alpha}) \subseteq \univ\setminus\kappa[\alpha]$, it follows that ${\sigma_\mathcal{M}}(\widetilde{\alpha}) = \kappa[\alpha]$.
    This gives 1.
    Finally, we note that our compatibility condition in Eq.~\ref{eq:compatibility} makes $\sigma_\mathcal{M}(\widetilde{\alpha})$ equal to an ideal in $(\univ,\leadsto)$.
    Since we've shown $\kappa[\alpha] = \sigma_\mathcal{M}(\widetilde{\alpha})$, this gives 2.

    This description of the flat supports is convenient for showing they are closed under intersection.
    In particular, since ideals of a poset are always closed under intersection, it follows that $\kappa(F)\cap\kappa(F')$ is an ideal on $(\univ,\leadsto)$ for all $F,F'\in\mathcal{L}(\mathcal{M},\leadsto)/\mathord{\sim}$.
    It is also true that $\kappa(F\sqcap F')$ and $\kappa(F)\cap\kappa(F')$ are closed in $\mathcal{M}$.
    Now, observe that the inclusion $\kappa(F\sqcap F') \subseteq \kappa(F) \cap \kappa(F')$ is always true by $(\glang/\mathord{\sim},\sqsubseteq)\cong(\kappa(\glang/\mathord{\sim}),\subseteq)$, so for 2. to hold we need only verify the reverse inclusion.
    By way of contradiction, suppose there exists $z \in (\kappa(F) \cap \kappa(F'))\setminus\kappa(F\sqcap F')$.
    Without loss of generality, let $z$ be a minimal letter in this difference w.r.t. `$\leadsto$.'
    Because $\kappa(F\sqcap F')$ is an ideal, minimality implies $y \in \kappa(F\sqcap F')$ for all $y\leadsto z$.
    Therefore, for all $\mu \in F\sqcap F'$, every prerequisite of $z$ is covered by the span ${\sigma_\mathcal{M}}(\widetilde{\mu})$ since ${\sigma_\mathcal{M}}(\widetilde{\mu}) = \kappa(F\sqcap F')$ by 1.
    This makes $z$ a continuation of words in $F\sqcap F'$, hence,
    \begin{equation}\label{eq:flkcon1}
        F\sqcap F' \prec [\mu z].
    \end{equation}
    But, $z\in\kappa(F)\cap \kappa(F')$ implies $[\mu z]$ is lesser than both $F$ and $F'$.
    One (of many) ways to see why is because if we select a word $\alpha \in \mathcal{L}(\mathcal{M},\leadsto)/\mu$ making $\mu\alpha \in F$, then by exchange there is an $\alpha$-subword $\alpha'$ of length $|\alpha|-1$ such that $\mu z \alpha'$ is feasible.
    This construction makes both $\widetilde{\mu z \alpha}' \subseteq \kappa(F)$, and the flats $F$ and $[\mu z \alpha']$ the same height in $(\glang/\mathord{\sim},\sqsubseteq)$.
    So, by $(\glang/\mathord{\sim},\sqsubseteq)\cong(\kappa(\glang/\mathord{\sim}),\subseteq)$, this makes $\mu z \alpha' \in F$.
    A symmetric argument applies to show $[\mu z]\sqsubseteq F'$ as well.
    Therefore,
    \begin{equation}\label{eq:flkcon2}
        [\mu z]\sqsubseteq F\sqcap F'.
    \end{equation}
    Yet, Eq.s~\ref{eq:flkcon1} and~\ref{eq:flkcon2} are contradictory, so $\kappa(F) \cap \kappa(F')=\kappa(F\sqcap F')$.
    This shows 3.
\end{proof}

\begin{figure}
    \centering
    \scalebox{.6}{\begin{tikzpicture}
    \tikzstyle{S}=[rectangle, draw=black, rounded corners=5pt,align=center,fill=white]
    \node[S] (e) {$\{\epsilon\}$};
    \node[S] (t) [above=of e, yshift=-10pt] {$\{t_1,t_2\}$};
    \node[S] (a1) [above left=of t, yshift=-5pt] {$\{t_1a_1,t_2a_1\}$};
    \node[S] (a2) [above right=of t, yshift=-5pt] {$\{t_1a_2,t_2a_2\}$};
    \node (inv) [above=of t] {};
    \node[S] (aa) [above=of inv,yshift=15pt] {$\{t_1a_1a_2, t_1a_2a_1,$\\$t_2a_1a_2, t_2a_2a_1\}$};
    \node[S] (a2s3) [right=of aa,xshift=-25pt,yshift=-25pt] {$\{t_1a_2s_3,t_2a_2s_3\}$};
    \node[S] (a2s4) [right=of a2s3,xshift=-20pt] {$\{t_1a_2s_4,t_2a_2s_4\}$};
    \node[S] (a1s2) [left=of aa,xshift=25pt,yshift=-25pt] {$\{t_1a_1s_2,t_2a_1s_2\}$};
    \node[S] (a1s1) [left=of a1s2,xshift=20pt] {$\{t_1a_1s_1,t_2a_1s_1\}$};
    \node[S] (top) [above=of aa,yshift=-10pt] {$\{t_1a_1 s_1 s_2, t_1a_1s_2s_1, t_1a_2s_3s_4, t_1a_2 s_4s_3, t_1a_2s_4s_1,t_1a_2s_4s_2,$\\$ t_1a_1a_2s_1, t_1a_1a_2s_2,t_1a_1a_2s_3, t_1a_2a_1s_1,t_1a_2a_1s_2,t_1a_2a_1s_3,$\\$t_2a_1 s_1 s_2, t_2a_1s_2s_1, t_2a_2s_3s_4, t_2a_2 s_4s_3, t_2a_2s_4s_1,t_2a_2s_4s_2,$\\$ t_2a_1a_2s_1, t_2a_1a_2s_2,t_2a_1a_2s_3, t_2a_2a_1s_1,t_2a_2a_1s_2,t_2a_2a_1s_3\}$};
    \node (label) [below=of e,yshift=20pt] {\Large $(\glang/\mathord{\sim},\sqsubseteq)$};

    \foreach \from/\to in {e/t, t/a1, t/a2, a1/a1s1, a1/a1s2, a1/aa, a2/aa, a2/a2s3, a2/a2s4, aa/top, a1s1/top, a1s2/top, a2s3/top, a2s4/top}
    \draw (\from) -- (\to);

    \node[S] (ef) [right=of e, xshift=290pt, yshift=-67.5pt] {$\varnothing$};
    \node[S] (tf) [above=of ef, yshift=-10pt] {$\{t_1,t_2\}$};
    \node[S] (a1f) [above left=of tf, yshift=-5pt] {$\{t_1,t_2,a_1\}$};
    \node[S] (a2f) [above right=of tf, yshift=-5pt] {$\{t_1,t_2,a_2\}$};
    \node (invf) [above=of tf] {};
    \node[S] (aaf) [above=of invf,yshift=15pt] {$\{t_1,t_2,a_1,a_2\}$};
    \node[S] (a2s3f) [right=of aaf,xshift=-25pt,yshift=-20pt] {$\{t_1,t_2,a_2,s_3\}$};
    \node[S] (a2s4f) [right=of a2s3f,xshift=-20pt] {$\{t_1,t_2,a_2,s_4\}$};
    \node[S] (a1s2f) [left=of aaf,xshift=25pt,yshift=-20pt] {$\{t_1,t_2,a_1,s_2\}$};
    \node[S] (a1s1f) [left=of a1s2f,xshift=20pt] {$\{t_1,t_2,a_1,s_1\}$};
    \node[S] (topf) [above=of aaf,yshift=-10pt] {$\{t_1,t_2,a_1,a_2,s_1,s_2,s_3,s_4\}$};
    \node (labelf) [below=of ef,yshift=20pt] {\Large $(\kappa(\glang/\mathord{\sim}),\subseteq)$};

    \foreach \from/\to in {ef/tf, tf/a1f, tf/a2f, a1f/a1s1f, a1f/a1s2f, a1f/aaf, a2f/aaf, a2f/a2s3f, a2f/a2s4f, aaf/topf, a1s1f/topf, a1s2f/topf, a2s3f/topf, a2s4f/topf}
    \draw (\from) -- (\to);
   
    \begin{pgfonlayer}{bg}
    \foreach \from/\to in {e/ef, t/tf, a1/a1f, a2/a2f, aa/aaf, a2s3/a2s3f, a2s4/a2s4f,a1s2/a1s2f,a1s1/a1s1f,top/topf}
    \draw[latex-latex, thick, dotted, gray] (\from) -- (\to);
    \end{pgfonlayer}

%
 

\end{tikzpicture}}
    \caption{The lattice of flats $(\glang/\mathord{\sim},\sqsubseteq)$ and flat supports $(\kappa(\glang/\mathord{\sim}),\subseteq)$ associated with the MPS from Figure~\ref{fig:bridge}. Observe, the meet operation of $(\kappa(\glang/\mathord{\sim}),\subseteq)$ is intersection, so the flat supports are closed under intersection. Moreover, every flat support is an ideal of $(\univ,\leadsto)$. Now, examine $\{a_1,a_2\}$. We see $\kappa[t_1 a_1 a_2] = \{t_1,t_2,a_1,a_2\}$ is the smallest flat support containing $\{a_1,a_2\}$. Hence, $\protect\gtest(\{a_1,a_2\})$ equals the height of $[t_1 a_1 a_2]$ in $(\glang/\mathord{\sim},\sqsubseteq)$ which is 3. This also makes $\cl(\{a_1,a_2\}) = \{t_1,t_2,a_1,a_2\}$.}\label{fig:spg-example}
\end{figure}
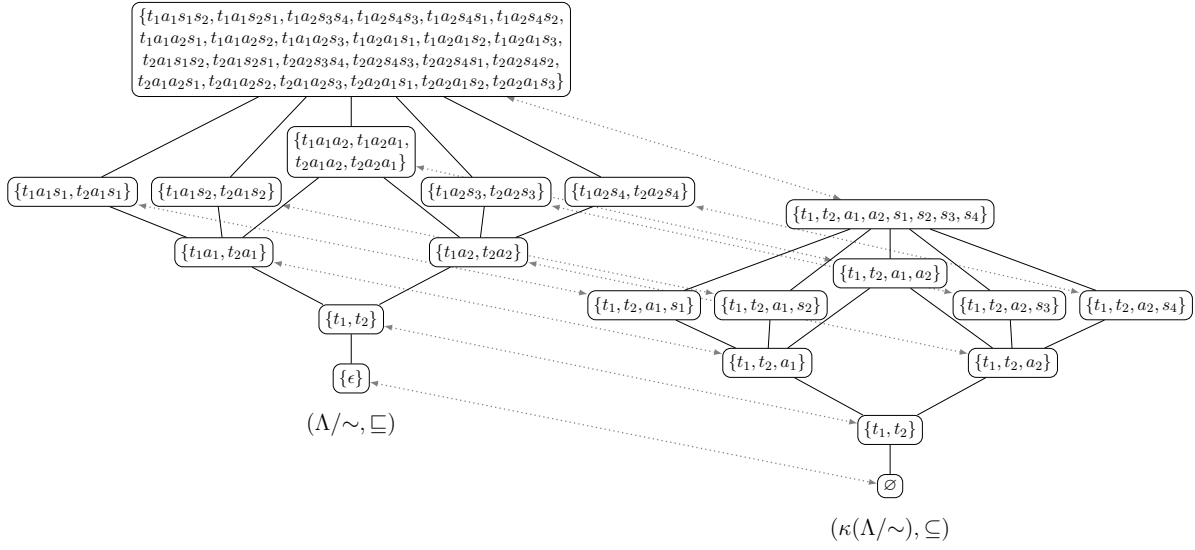

Therefore, by Lemma~\ref{lem:int}.3 and Proposition~\ref{prop:opt+int}, every MPS is a strong polymatroid greedoid by the description from \cite{streit2024polymatroid}.
To give some concrete intuition about the greatest representation and closure structure, in Figure~\ref{fig:spg-example} we show the lattice of flats $(\glang/\mathord{\sim},\sqsubseteq)$, lattice of flat supports $(\kappa(\glang/\mathord{\sim}),\subseteq)$, and their correspondence from the example in Figure~\ref{fig:bridge}.
For logical completeness, we give a construction encoding every strong polymatroid greedoid as an MPS in what follows.

\begin{theorem}\label{thm:spg-iff-mps}
    Let $\glang$ be a greedoid on the alphabet $\univ$.
    There exists a matroid $(\univ,\mathcal{M})$ and compatible poset $(\univ,\leadsto)$ such that $\glang$ equals the words of the matroidal prerequisite system $\mathcal{L}(\mathcal{M},\leadsto)$ if and only if $\glang$ is a strong polymatroid greedoid.
\end{theorem}
\begin{proof}
    As identified, Lemma~\ref{lem:int}.3 and Proposition~\ref{prop:opt+int} already show one direction, so we need only verify that every strong polymatroid greedoid can be encoded as an MPS.
    Let $\glang$ be a strong polymatroid greedoid with greatest representation $\gtest$ and closure operator $\cl$.
    Consider the 1-truncation,
    \[
        \gtest^{(1)}(X) \triangleq \min_{Y\subseteq X} \gtest(Y) + |X\setminus Y|.
    \]
    It is known that $\gtest^{(1)}$ is always a matroid rank function; see \cite{korte2012greedoids} for details on $k$-truncations of polymatroids.
    So, we let $\mathcal{M}$ be the matroid with rank $\gtest^{(1)}$.
    Now, define the poset $(\univ, \leadsto)$ via the ordering of the closures of singleton sets like,
    \begin{equation}\label{eq:order}
        y\leadsto z \equiv \cl(y) \subset \cl(z).
    \end{equation}
    We verify the compatibility condition.
    Let $\alpha \in \mathcal{L}(\mathcal{M},\leadsto)$, $y\leadsto z$, and $z \in \sigma_\mathcal{M}(\widetilde{\alpha})$.
    Assume $z \notin \widetilde{\alpha}$ and $z \in \widetilde{\alpha}$ implies $y \in \sigma_\mathcal{M}(\widetilde{\alpha})$ through Eq.~\ref{eq:feasibility}.
    This implies that there exists a minimizer $Y$ achieving $\gtest(Y) + |\widetilde{\alpha}\setminus Y| = \gtest^{(1)}(\widetilde{\alpha})$ with $(\gtest/Y)(z) = 0$.
    Since $\cl = \sigma_{\gtest}$, Eq.~\ref{eq:order} shows $(\gtest/Y)(y) = 0$ as well, so,
    \[
        \gtest^{(1)}(\widetilde{\alpha} + y)\geq\gtest^{(1)}(\widetilde{\alpha}) = \gtest(Y) + |\widetilde{\alpha}\setminus Y| \geq \gtest(Y + y) + |\widetilde{\alpha}\setminus (Y+y)| \geq \gtest^{(1)}(\widetilde{\alpha} + y).
    \]
    It follows that $y \in \sigma_\mathcal{M}(\widetilde{\alpha})$, so compatibility is satisfied.
    Hence, $\mathcal{L}(\mathcal{M},\leadsto)$ is an MPS.

    We claim $\glang = \mathcal{L}(\mathcal{M},\leadsto)$.
    First we verify $y_1\ldots y_\ell \in \glang$ implies $y_1\ldots y_\ell \in \mathcal{L}(\mathcal{M},\leadsto)$.
    By way of contradiction, suppose $\{y_1,\ldots,y_\ell\}$ is not independent.
    Then $\gtest^{(1)}(\{y_1,\ldots,y_\ell\}) < \ell$.
    Let $Z\subseteq \{y_1,\ldots,y_\ell\}$ be a minimizer achieving,
    \[
        \gtest^{(1)}(\{y_1,\ldots,y_\ell\}) = \gtest(Z) + |\{y_1,\ldots,y_\ell\}\setminus Z|.
    \]
    Because $y_1\ldots y_\ell \in \glang$, the law of diminishing returns shows,
    \[
        (\gtest/Z\cup\{y_1,\ldots,y_{i-1}\})(y_i) \leq 1,\quad\forall i \in \{1,\ldots, \ell\}.
    \]
    Therefore, as $Z\subseteq \{y_1,\ldots,y_\ell\}$,
    \[
        \ell = \gtest(\{y_1,\ldots,y_\ell\}) = \gtest(Z) + \sum_{i=1}^\ell(\gtest/Z\cup\{y_1,\ldots,y_{i-1}\})(y_i) \leq \gtest(Z) + |\{y_1,\ldots,y_\ell\}\setminus Z| < \ell,
    \]
    a contradiction.
    It follows that $\{y_1,\ldots,y_\ell\} \in \mathcal{M}$.
    Now, fix $i \in \{1,\ldots,\ell\}$ and a prerequisite $x \leadsto y_i$.
    Assume $x \notin \sigma_\mathcal{M}(\{y_1,\ldots,y_{i-1}\})$.
    Because matroids are subset-closed families, $\{y_1,\ldots,y_{i}\}$ is independent as well.
    Therefore,
    \[
        \gtest^{(1)}(\{y_1,\ldots,y_{i}\}) = \gtest(\{y_1,\ldots,y_{i}\}) = \gtest(\{y_1,\ldots,y_{i},x\}) \geq \gtest^{(1)}(\{y_1,\ldots,y_{i},x\}) \geq \gtest^{(1)}(\{y_1,\ldots,y_{i}\}),
    \]
    which implies $x \in \sigma_\mathcal{M}(\{y_1,\ldots,y_{i}\})$.
    So, by the Steinitz-Maclane exchange property, $x$ and $y_i$ are $\{y_1,\ldots,y_{i-1}\}$-substitutable.
    As our choice of index $i$ was arbitrary, this makes $y_1\ldots y_\ell \in \mathcal{L}(\mathcal{M},\leadsto)$.

    Now suppose $y_1\ldots y_\ell \notin \glang$.
    Then, there exists $i \in \{1,\ldots,\ell\}$ with $(\gtest/\{y_1,\ldots,y_{i-1}\})(y_i) \neq 1$; fix $i$ to be the first such index.
    If $(\gtest/\{y_1,\ldots,y_{i-1}\})(y_i) = 0$, then $\{y_1,\ldots, y_\ell\}$ cannot be independent since we would have $\gtest^{(1)}(\{y_1,\ldots,y_{\ell}\}) < \ell$ by the fact that $\gtest^{(1)}(\{y_1,\ldots,y_{i}\}) < i$ and $\gtest^{(1)}$ is subcardinal (so marginal returns are bounded by one).
    If $(\gtest/\{y_1,\ldots,y_{i-1}\})(y_i) > 1$, then $\cl(y_i)\not\subseteq \cl(\{y_1,\ldots,y_{i-1}\})$, and by way of semimodularity $\cl(y_i)$ cannot cover any flat support (from $\glang$) in $\ideal \cl(\{y_1,\ldots,y_{i-1}\})$ on $(\kappa(\glang/\mathord{\sim}),\subseteq)$.
    So, there exists $x \leadsto y_i$ such that,
    \begin{equation}\label{eq:cl-cont}
        \cl(\{y_1,\ldots,y_{i-1}\}) \prec \cl(\{y_1,\ldots,y_{i-1},x\}) \subset \cl(\{y_1,\ldots,y_i\}).
    \end{equation}
    Before continuing, take note that our choice of $i$ makes $y_1\ldots y_{i-1}\in\glang$.
    Moreover, since $\gtest$ grades the $\cl$-image ordered by containment, 
    $$\cl(\{y_1,\ldots,y_{i-1}\}) \prec \cl(\{y_1,\ldots,y_{i-1},x\})\implies(\gtest/\{y_1,\ldots,y_{i-1}\})(x) = 1.$$
    So $y_1\ldots y_{i-1}x\in\glang$, which means $\{y_1,\ldots,y_{i-1},x\}\in\mathcal{M}$ by our prior argument.
    This means,
    \begin{equation}\label{eq:first-cont}
        Z \subseteq \{y_1,\ldots,y_{i-1},x\} \implies \gtest(Z) + |\{y_1,\ldots,y_{i-1},x,y_i\}\setminus Z| > |\{y_1,\ldots,y_{i-1},x\}|.
    \end{equation}
    Moreover, Eq.~\ref{eq:cl-cont} implies $(\gtest/\{y_1,\ldots,y_{i-1},x\})(y_i) \geq 1$.
    Recalling $\gtest(Z) = \gtest(Z-y_i) + (\gtest/Z-y_i)(y_i)$ whenever $y_i \in Z$, by the law of diminishing returns one observes,
    \begin{equation}\label{eq:second-cont}
        Z \subseteq \{y_1,\ldots,y_{i-1},x,y_i\}\text{ and }y_i \in Z \implies \gtest(Z) + |\{y_1,\ldots,y_{i-1},x,y_i\}\setminus Z| > |\{y_1,\ldots,y_{i-1},x\}|,
    \end{equation}
    as well.
    Combining Eq.s~\ref{eq:first-cont} and~\ref{eq:second-cont} shows $y_i \notin \sigma_\mathcal{M}(\{y_1,\ldots,y_{i-1},x\})$, and so by contrapositive of the Steinitz-Maclane exchange property it is the case that $x \notin \sigma_\mathcal{M}(\{y_1,\ldots,y_i\})$.
    So, $x\leadsto y_i$ while $x$ is neither covered by $\sigma_\mathcal{M}(\{y_1,\ldots,y_{i-1}\})$ or $\{y_1,\ldots,y_{i-1}\}$-substitutable with $y_i$.
    And so, $y_1\ldots y_\ell\notin\mathcal{L}(\mathcal{M},\leadsto)$ in the case of $(\gtest/\{y_1,\ldots,y_{i-1}\})(y_i) >1$ as well.
    This concludes the proof.
\end{proof}

\subsection{Computing the Greatest Representation}\label{subsec:mps-gr}

Our algorithms use the greatest representation $\gtest$.
To ensure this is well-founded, we show $\gtest$ can be evaluated via a single rank query on the underlying matroid.
Specifically, the greatest representation of an MPS $\mathcal{L}(\mathcal{M},\leadsto)$ is the composition of the matroid rank and ideal operator on $(\univ,\leadsto)$.
Note that ideals can be computed efficiently in $\poly(|\univ|)$-time by repeated ordered comparisons.

\begin{proposition}\label{prop:rideal}
    Let $\mathcal{L}(\mathcal{M},\leadsto)$ be an MPS.
    Then, $\gtest = r_\mathcal{M}\circ\ideal$.
\end{proposition}
\begin{proof}
    Fix $X\subseteq\univ$.
    We verify $\gtest(X) = r_\mathcal{M}(\ideal X)$.
    Perform the following procedure: Begin with a word $\alpha = \epsilon$.
    Then, in each step, select a smallest $z \in \ideal X\setminus\widetilde{\alpha}$ satisfying $r_\mathcal{M}(\widetilde{\alpha} + z) > r_\mathcal{M}(\widetilde{\alpha})$.
    Terminate once no such $z$ exists.
    See that $\alpha \in \mathcal{L}(\mathcal{M},\leadsto)$ because, fixing $z$ to a selection made in any round, minimality makes every $y \leadsto z$ spanned by the round's value of $\alpha$.
    Moreover, it is the case that $(r_\mathcal{M}/\widetilde{\alpha})(\ideal X) = 0$ upon termination.
    So, with $\alpha$ being the value at the end of the final round, we have $\ideal X \subseteq {\sigma_\mathcal{M}}(\widetilde{\alpha})$.
    Yet combining this with $\widetilde{\alpha}\subseteq\ideal X$ makes,
    \[
        r_\mathcal{M}(\widetilde{\alpha}) \leq r_\mathcal{M}(\ideal X) \leq r_\mathcal{M}(\widetilde{\alpha}),
    \]
    so $r_\mathcal{M}(\ideal X) = |\alpha|$ since $\alpha \in \mathcal{L}(\mathcal{M},\leadsto)$ implies $\widetilde{\alpha} \in \mathcal{M}$.
    Now, take note that $\ideal X \subseteq {\sigma_\mathcal{M}}(\widetilde{\alpha})$ also implies that no $\beta \in \mathcal{L}(\mathcal{M},\leadsto)/\alpha$ satisfies $\widetilde{\beta}\cap\ideal X\neq \varnothing$, as the supports of feasible words must be independent.
    This makes $(\gtest/\widetilde{\alpha})(y) \neq 1$ for all $y \in \ideal X$ in particular.
    By optimism (see Proposition~\ref{prop:opt+int}) it follows that there exists an $\alpha$-prefix $\alpha^{(y)}z$ for each $y \in \ideal X$ such that $\alpha^{(y)}y$ is feasible.
    So, $(\gtest/\widetilde{\alpha}^{(y)})(y) = 1$.
    But by the law of diminishing returns and $(\gtest/\widetilde{\alpha})(y) \neq 1$ it follows that,
    \[
        1 = (\gtest/\widetilde{\alpha}^{(y)})(y) > (\gtest/\widetilde{\alpha})(y),\quad\forall y \in \ideal X,
    \]
    and so as $\gtest$ is integral we have $(\gtest/\widetilde{\alpha})(y) = 0$ for all $y \in \ideal X$.
    This makes $\ideal X \subseteq \cl(\widetilde{\alpha})$.
    Therefore,
    \[
        \gtest(\widetilde{\alpha}) \leq \gtest(\ideal X) \leq \gtest(\widetilde{\alpha}),
    \]
    and so $\gtest(\ideal X) = \gtest(\widetilde{\alpha})$.
    But, because flat supports are ideals by Lemma~\ref{lem:int}.2, we see $\ideal X \subseteq \cl(X)$, and so $\gtest(X) \leq \gtest(\ideal X) \leq \gtest(X)$ by $\cl = \sigma_{\gtest}$.
    Therefore, $\gtest(X) = \gtest(\widetilde{\alpha})$,
    and since $\gtest$ is a representation we have $\gtest(\widetilde{\alpha}) = |\alpha|$.
    Hence we've found $r_\mathcal{M}(\ideal X) = |\alpha|$ and $\gtest(X) = |\alpha|$.
\end{proof}

\subsection{Projection and Fill-In}\label{subsec:fill}

In Sections~\ref{sec:scaffold} we present algorithms based on projecting a relaxed solution back onto the words of an MPS, while in Section~\ref{sec:mag} our algorithms build words through repeated insertions eventually converging to a feasible word.
Both ideas require a description of what sets of letters, which might not support a feasible word, can be ``filled'' into a feasible word.
In the following, we show that words $x_1\ldots x_\ell$ such that $(\gtest/\{x_1,\ldots,x_{i-1}\})(x_i) >0$ for all $i \in \{1,\ldots, \ell\}$ satisfy this.
We note this is a special property of strong polymatroid greedoids.

\begin{lemma}\label{lem:projection}
    Suppose $\glang$ is a strong polymatroid greedoid.
    Then, for all simple $x_1\ldots x_k \in \univ^*$, if $(\gtest/\{x_1,\ldots,x_{i-1}\})(x_i) > 0$ for all $i \in \{1,\ldots,k\}$, there exists a set of letters $\{y_{i,1},\ldots,y_{i,\ell_i}\}$ with $\ell_i =(\gtest/\{x_1,\ldots,x_{i-1}\})(x_i) - 1$ inducing the feasible word,
    \begin{equation}\label{eq:marriage}
        y_{1,1}\ldots y_{1,\ell_1} x_1 y_{2,1}\ldots y_{2,\ell_2} x_2 \ldots y_{k,1}\ldots y_{k,\ell_k}x_k \in \glang.
    \end{equation}
\end{lemma}
\begin{proof}
    By $(\gtest/\{x_1,\ldots,x_{i-1}\})(x_i) > 0$ for all $i \in \{1,\ldots,k\}$, the existence of the strictly ascending,
    \begin{equation}\label{eq:strict-chain}
        \varnothing \subsetneq \cl(\{x_1\}) \subsetneq \cl(\{x_1,x_2\}) \subsetneq \ldots \subsetneq \cl(\{x_1,\ldots,x_k\}),
    \end{equation}
    of flat supports in $(\kappa(\glang/\mathord{\sim}),\subseteq)$ follows.
    Then, to obtain a feasible word like Eq.~\ref{eq:marriage}, select letters $\{y_{i,1},\ldots,y_{i,\ell_i}\}$ for each $x_i$ inducing a covering chain beginning at $\cl(\{x_1,\ldots,x_{i-1}\})$ and ending at $\cl(\{x_1,\ldots,x_i\})$ of the form,
    \begin{multline}\label{eq:covering}
        \cl(\{x_1,\ldots,x_{i-1}\}) \prec \cl(\{x_1,\ldots,x_{i-1},y_{i,1}\}) 
        \prec\ldots \prec \cl(\{x_1,\ldots,x_{i-1},y_{i,1},\ldots,y_{i,\ell_i}\}),\\ \prec \cl(\{x_1,\ldots,x_{i-1},y_{i,1},\ldots,y_{i,\ell_i},x_i\}) = \cl(\{x_1,\ldots,x_i\}).
    \end{multline}
    This is always possible because $\cl$ is a closure operator.
    In particular, since $\sigma_{\gtest} = \cl$, having already selected $y_{i,1},\ldots,y_{i,j-1}$, one takes an arbitrary $y_{i,j}$ satisfying,
    \begin{equation}\label{eq:sel1}
        (\gtest/\{x_1,\ldots,x_{i-1},y_{i,1},\ldots,y_{i,j-1}\})(y_{i,j}) = 1,
    \end{equation}
    and,
    \begin{equation}\label{eq:sel2}
        \cl(\{x_1,\ldots,x_{i-1},y_{i,1},\ldots,y_{i,j-1}\} + y_{i,j}) \subseteq \cl(\{x_1,\ldots,x_i\}),
    \end{equation}
    whenever $(\gtest/\{x_1,\ldots,x_{i-1},y_{i,1},\ldots,y_{i,j-1}\})(x_i) \neq 1$, and terminates otherwise.
    A letter $y_{i,j}$ satisfying Eq.~\ref{eq:sel1} and~\ref{eq:sel2} must exist when $(\gtest/\{x_1,\ldots,x_{i-1},y_{i,1},\ldots,y_{i,j-1}\})(x_i) \neq 1$, since there is a flat support lesser than $\cl(\{x_1,\ldots,x_i\})$ covering $\cl(\{x_1,\ldots,x_{i-1},y_{i,1},\ldots,y_{i,j-1}\})$ in $(\kappa(\glang/\mathord{\sim}),\subseteq)$, and any letter in the difference of this cover and $\cl(\{x_1,\ldots,x_{i-1},y_{i,1},\ldots,y_{i,j-1}\})$ will satisfy Eq.~\ref{eq:sel1} and~\ref{eq:sel2}. 
    Now, recall that $\gtest(X)$ is defined as the height of the least flat support containing $X$ (recall flat supports are closed under intersection, so a least such flat is unique).
    So, $\gtest$ grades the $\cl$-image.
    Then, the covering chain structure from Eq.~\ref{eq:covering} and $\sigma_{\gtest} = \cl$ makes,
    \begin{equation}\label{eq:slsls}
        (\gtest/\{x_1,\ldots,x_{i-1},y_{i,1},\ldots,y_{i,j-1}\})(y_{i,j}) = 1,\quad\forall j \in \{1,\ldots,\ell_i\}.
    \end{equation}
    This means, for example, that $y_{1,1}\ldots y_{1,\ell_1} x_1 \in\glang$ by the definition of representation in polymatroid greedoids.
    Then, since selecting according to Eq.~\ref{eq:covering} makes $\cl(\{y_{1,1},\ldots,y_{1,\ell_1},x_1\}) = \cl(\{x_1\})$, applying the same argument using Eq.~\ref{eq:slsls} for $i=2$ shows $y_{1,1}\ldots y_{1,\ell_1} x_1y_{2,1}\ldots y_{2,\ell_2} x_2 \in\glang$.
    One repeats this procedure up to $i = k$ to get the result.
\end{proof}

\section{Scaffolding Solutions via Polymatroid Relaxations}\label{sec:scaffold}

Observe, every polymatroid greedoid is such that the characteristic vectors of its basic words are base vectors of any of its polymatroid representations.
To see why, note that the variant of Edmonds' algorithm we gave in Algorithm~\ref{alg:edmonds} is optimal for linear optimization over the base polytope, and so always returns an extreme point.
Let $y_1\ldots y_R \in\mathcal{L}(\mathcal{M},\leadsto)$ be any basic word of an MPS with greatest representation $\gtest$.
Then, a weight function $w:\univ\to\mathbb{R}$ such that $w(y_i) = R-(i-1)$ for all $i \in \{1,\ldots,R\}$ and $w(z) = 0$ for $z \notin\{y_1,\ldots y_R\}$ will make Edmonds' algorithm examine coordinates in the order dictated by the basic word.
Because $\gtest$ is a representation, we see that $(\gtest/\{y_1,\ldots,y_{i-1}\})(y_i) = 1$ for all $i$, and Lemma~\ref{lem:projection} certifies that this ends at $y_R$, as the existence of $z \in \univ$ with $(\gtest/\{y_1,\ldots,y_R\})(z) >0$ would imply $y_1\ldots y_R$ is not basic.
Therefore, Edmonds' algorithm will increase each coordinate from $\{y_1,\ldots,y_R\}$ to exactly one and terminate on the characteristic vector of the support of the basic word.

In what follows we will examine nonnegative additive maximization and monotone submodular maximization algorithms based on using the greatest representation as a relaxation.
Specifically, Section~\ref{subsec:scaffold-am} will first describe a method for additive maximization in Algorithm~\ref{alg:greedy-scaffold} which is based on computing a relaxed solution via Edmonds' algorithm and projecting onto a basic word using Lemma~\ref{lem:projection}.
In Section~\ref{subsec:scaffold-sub} we then approach submodular maximization by using the continuous greedy algorithm \cite{calinescu2011maximizing} to compute a set of weights forming an additive underestimation of the objective, which we then give to our Algorithm~\ref{alg:greedy-scaffold} to compute a solution on the MPS.
The resulting approximation guarantees will vary in largest matroid rank of a principal ideal $\Delta$, which is intuitively the distance of the MPS from one encoding a matroid as discussed in Section~\ref{sec:intro}.

\subsection{Additive Maximization}\label{subsec:scaffold-am}

Let $w:\univ\to\mathbb{R}_+$ be a nonnegative weight function.
Our task is to find a feasible word $\alpha$ of a MPS $\mathcal{L}(\mathcal{M},\leadsto)$ maximizing $\alpha \mapsto \sum_{y\in\widetilde{\alpha}}w(y)$.
Let $v^\star\in\mathbb{R}^\univ$ be base vector of $\gtest$ computed by Edmonds' algorithm, and sort the nonzero coordinates $\{x_1,x_2,\ldots\} =\supp(v^*)$ so that $w(x_1) \geq w(x_2) \geq \ldots$.
This is, we sort the support of $v^\star$ in order of the coordinate increases made by Algorithm~\ref{alg:edmonds}.
Since $x_i \in \supp(v^\star)$, it follows that,
\[
     (\gtest/\{x_1,\ldots,x_{i-1}\})(x_i) > 0,\quad\forall i \in\{1,\ldots,|\supp(v^\star)|\},
\]
which satisfies the premise of Lemma~\ref{lem:projection}.
As a result, we can perform a fill-in procedure like in the proof of Lemma~\ref{lem:projection} to obtain a feasible word of the MPS containing $\supp(v^\star)$ in its support.

\begin{algorithm}[t]
    \caption{Scaffolded Greedy with Arbitrary Fill-In}\label{alg:greedy-scaffold}
    \DontPrintSemicolon
    
    $v^\star \gets$ solution from Edmonds' algorithm on $\gtest$ with weights $w$\;\label{al:ed}
    select $x_1\ldots x_k \in \operatorname{supp}(v^\star)!$ with $w(x_1) \ge w(x_{2}) \ge \ldots \ge w(x_{k})$\;\label{al:sort}
    \;
    $\alpha\gets\epsilon$\;
    \;\label{al:fill-start}
    \For{$i\in\{1,\ldots,k\}$}
    {
        \While{$(\gtest/\widetilde{\alpha})(x_i) > 1$}
        {
            \mbox{select $y\in\univ$ with \!$\left(\gtest/\widetilde{\alpha}\right)\!(y) = 1$ and $\widetilde{\alpha} + y \subseteq \cl(\{x_{1},\ldots,x_{i}\})$}\;

            $\alpha \gets \alpha y$\;
        }
        \;
        $\alpha\gets \alpha x_i$\;\label{al:fill-end}
    }
    \;
    
    \Return{$\alpha$}\;

\end{algorithm}

This idea is shown in Algorithm~\ref{alg:greedy-scaffold}.
We first compute a ``scaffold'' $v^\star$ for the solution using Edmonds' greedy algorithm on $\gtest$ in line~\ref{al:ed}, whose support we then sort in line~\ref{al:sort}.
Then, we use the same arbitrary fill-in scheme presented in the proof of Lemma~\ref{lem:projection}, in particular Eq.s~\ref{eq:sel1} and~\ref{eq:sel2}, to obtain a feasible word $\alpha$ such that $\supp(v^\star) \subseteq \widetilde{\alpha}$ in lines~\ref{al:fill-start} through~\ref{al:fill-end}.
Note that the while loop cannot terminate with
\((\gtest/\widetilde{\alpha})(x_i)=0\): if a selected fill-in
letter \(y\) with \((\gtest/\widetilde{\alpha})(y)=1\) made
\(x_i\in \cl(\widetilde{\alpha}+y)\), then we would have to have 
\((\gtest/\widetilde{\alpha})(x_i)\le 1\) which contradicts the loop condition.

Recall, in Section~\ref{sec:intro} we defined $\Delta$ as the maximum matroid rank of a principal ideal of $(\univ,\leadsto)$.
By the law of diminishing returns, it follows that each coordinate of $v^\star$ is at most $\Delta$.
When comparing $v^\star$ and $\chi_{\widetilde{\alpha}}$, it follows that any entry in $v^\star$ is at most a $\Delta$-factor smaller in $\chi_{\widetilde{\alpha}}$.
This idea shows Algorithm~\ref{alg:greedy-scaffold} computes a $\Delta$-approximation.

\begin{theorem}\label{thm:scaffold-approx-add}
    Let $\mathcal{L}(\mathcal{M},\leadsto)$ be a matroidal prerequisite system, $w:\univ\to\mathbb{R}_+$ a nonnegative weight function,
    and $\mathrm{OPT} =\max_{\alpha \in \mathcal{L}(\mathcal{M},\leadsto)} \sum_{y \in \widetilde{\alpha}} w(y)$.
    Algorithm~\ref{alg:greedy-scaffold} outputs $\alpha\in\mathcal{L}(\mathcal{M},\leadsto)$ with,
    \[
        \sum_{y \in \widetilde{\alpha}}w(y) \geq \frac{1}{\Delta}\cdot\mathrm{OPT}.
    \]
\end{theorem}
\begin{proof}
    Firstly, one sees that the fill-in (lines \ref{al:fill-start}-\ref{al:fill-end}) is done in the same way as the proof of Lemma~\ref{lem:projection}, so it is clear from our discussion at the start of the section that the resulting $\alpha$ is feasible.
    Moreover, it is basic as $\gtest(\supp(v^\star)) = \gtest(\univ)$ and $\cl(\widetilde{\alpha}) = \cl(\supp(v^\star))$ makes $\gtest(\widetilde{\alpha}) = \gtest(\univ)$.
    For each $x_{i}\in\supp(v^\star)$, by definition of Edmonds' algorithm (see Algorithm~\ref{alg:edmonds}),
    \[
        v^\star(x_{i}) = (\gtest/\{x_{1},\ldots,x_{i-1}\})(x_{i}).
    \]
    Therefore, for any choice of $y \in \supp(v^\star)$, we have $v^\star(y) \leq \gtest(y)$ by the law of diminishing returns.
    Of course $\gtest(y) \leq \Delta$, and so,
    \begin{equation}\label{eq:sc-relax-bd}
        \chi_{\widetilde{\alpha}}(y) \geq \frac{1}{\Delta}\cdot v^\star(y),\quad\forall y \in \supp(v^\star).
    \end{equation}
    Then, let $\alpha_\mathrm{OPT}$ be an optimal word.
    Recall by our discussion at the beginning of Section~\ref{sec:scaffold} that $\chi_{\widetilde{\alpha}_\mathrm{OPT}}$ is an extreme point of the base polytope, and so as Edmonds' algorithm is optimal over this constraint set we see,
    \begin{equation}\label{eq:sc-opt-bd}
        \mathrm{OPT} = w^\intercal \chi_{\widetilde{\alpha}_\mathrm{OPT}} \leq w^\intercal v^\star.
    \end{equation}
    Combining Eq.s~\ref{eq:sc-relax-bd} and~\ref{eq:sc-opt-bd} together shows the desired $w^\intercal \chi_{\widetilde{\alpha}} \geq \Delta^{-1}\cdot w^\intercal v^\star \geq \Delta^{-1}\cdot\mathrm{OPT}$.
\end{proof}

\subsection{Monotone Submodular Maximization}\label{subsec:scaffold-sub}

For what comes, we require the following: Every integral polymatroid can be associated with a matroid via the \emph{natural matroid} construction \cite{helgason1974aspects}.
Specifically, define the ground set as,
$$E_{\gtest} \triangleq \bigcup_{y \in \univ}\{(y, 1),\ldots,(y,\gtest(y))\}.$$
This just introduces $\gtest(y)$ copies of each letter $y\in\univ$.
This is equivalently seen as a multiset, so
let the multiplicity vector of $S \subseteq E_{\gtest}$ be given by,
\[
    \mu_S(y) \triangleq |S\cap(\{y\}\times \{1,\ldots,\gtest(y)\})|.
\]
Then, a matroid $M(\gtest)$ is defined by letting $S$ be independent if and only if $\mu_S^\intercal\chi_X \leq \gtest(X)$ for all $X\subseteq\univ$.
The resulting rank function is,
\[
    r_{M(\gtest)}(S) \triangleq \min_{X\subseteq \univ} \gtest(X) + \sum_{y \in \univ\setminus X} \mu_S(y),
\]
which we note can be evaluated efficiently by submodular minimization algorithms \cite{lovasz1983submodular,schrijver2000combinatorial,iwata2001combinatorial,orlin2009faster} over $X\mapsto \gtest(X) + \sum_{y \in \univ\setminus X} \mu_S(y)$.
This is a standard construction; see \cite{helgason1974aspects,bonin2023natural} or Ch. 44 of \cite{schrijver2003combinatorial} for more discussion.

\begin{algorithm}[t]
    \caption{Scaffolding with Weights Computed from Continuous Greedy}\label{alg:scaffolded-submod}
    \DontPrintSemicolon
    $S \gets$ continuous greedy with rounding on objective $T\mapsto \sum_{k=1}^\Delta f(T(k))$ over $\left(E_{\gtest},M(\gtest)\right)$\;\label{al:cg}
    \;

    \ForAll{$k \in \{1,\ldots,\Delta\}$}
    {
        select any $y_1 y_2\ldots \in S(k)!$\;\label{al:lay-start}
        \;
        \ForAll{$i \in \{1,\ldots |S(k)|\}$}
        {
            $w^{(k)}(y_i) \gets (f/\{y_{1},\ldots,y_{i-1}\})(y_i)$\;
        }
        \ForAll{$z \notin S(k)$}
        {
            $w^{(k)}(z) \gets 0$\;\label{al:lay-stop}
        }
    }
    \; 
    $\overbar{w} \gets {\Delta}^{-1}\cdot\sum_{k=1}^\Delta w^{(k)}$\;\label{al:avg}
    $\alpha \gets$ scaffolded greedy on $\mathcal{L}(\mathcal{M},\leadsto)$ with weights $\overbar{w}$\tcp*{Using Algorithm~\ref{alg:greedy-scaffold}}\label{al:scaf}
    \; 
    \Return{$\alpha$}\;
\end{algorithm}

Now, let $f:2^\univ\to\mathbb{R}$ be a monotone submodular function.
Our next algorithm is based off of solving a monotone submodular maximization, using a derived objective function $g:2^{E_{\gtest}}\to\mathbb{R}$, over the natural matroid using the continuous greedy algorithm \cite{calinescu2011maximizing}.
The corresponding solution then defines an additive function underestimating the objective $f$, which we feed into Algorithm~\ref{alg:greedy-scaffold} to obtain a word from the MPS.
To define $g$, for $S\subseteq E_{\gtest}$ let its $k^\text{th}$ slice $S(k)$ be,
\[
    S(k) \triangleq \{y\in\univ \mid (y,k) \in S\}.
\]
Notice, as $\Delta$ is the maximum rank of a principal ideal and $\gtest = r_\mathcal{M}\circ\ideal$, it follows there are at most $\Delta$ slices to evaluate for any $S \subseteq E_{\gtest}$.
Then, define,
$$g:2^{E_{\gtest}}\to\mathbb{R}:S\mapsto \sum_{k=1}^\Delta f(S(k)).$$
Observe, $g$ is a sum of monotone submodular functions on disjoint domains (since the $k^\text{th}$ term in the sum varies only with sets in the intersection $E_{\gtest}\cap(\univ\times\{k\})$), and so is monotone submodular.

Algorithm~\ref{alg:scaffolded-submod} begins by obtaining an approximate solution $S$ to maximizing $g$ over the natural matroid $(E_{\gtest},M(\gtest))$ in line~\ref{al:cg}.
Then, in lines~\ref{al:lay-start} through~\ref{al:lay-stop} we compute a set of weights for each slice of $S$.
Specifically, an arbitrary permutation $y_1y_2\ldots$ is chosen from $S(k)!$ to define,
\begin{equation*}
    w^{(k)}(y_i) = (f/\{y_1,\ldots,y_{i-1}\})(y_i),\quad\forall i \in \{1,\ldots,|S(k)|\},
\end{equation*}
and $w(z) = 0$ for $z \notin S(k)$.
These weight functions are averaged together in line~\ref{al:avg} to obtain the weight function $\overbar{w}$.
This is, in some sense, an additive estimate of $f$ informed by a solution obtained from a proxy function.
We then obtain a word of the MPS in line~\ref{al:scaf} by running Algorithm~\ref{alg:greedy-scaffold} with the weights defined by $\overbar{w}$, and output that word.
The approximation guarantees are as follows.

\begin{theorem}\label{thm:scaffold-approx-sub}
    Let $\mathcal{L}(\mathcal{M},\leadsto)$ be a matroidal prerequisite system, $f:2^\univ\to\mathbb{R}$ a (normalized) monotone submodular function,
    and $\mathrm{OPT} = \max_{\alpha \in \mathcal{L}(\mathcal{M},\leadsto)}f(\widetilde{\alpha})$.
    Algorithm~\ref{alg:scaffolded-submod} outputs $\alpha\in\mathcal{L}(\mathcal{M},\leadsto)$ with,
    \[
        \mathbb{E}[f(\widetilde{\alpha})] \geq \frac{1 - 1/e - \delta}{\Delta^2}\cdot\mathrm{OPT},
    \]
    in randomized $\poly(|\univ|, 1/\delta)$-time for all $\delta > 0$.
\end{theorem}

The resulting guarantee is that of continuous greedy with pipage rounding \cite{calinescu2011maximizing} degraded by a $\Delta^2$-factor due to the maneuvers involved in translating its solution to a feasible word.
To assist the reader, we state the guarantees from \cite{calinescu2011maximizing}.

\begin{theorem}[\cite{calinescu2011maximizing}]
    Let $(E,\mathcal{M})$ be a matroid, $f:2^\univ\to\mathbb{R}$ a (normalized) monotone submodular function,
    and $\mathrm{OPT} = \max_{X\in\mathcal{M}}f(X)$.
    There exists an algorithm computing an independent set $X\in\mathcal{M}$ satisfying,
    \[
        \mathbb{E}[f(X)] \geq (1 - 1/e - \delta)\cdot\mathrm{OPT},
    \]
    in randomized $\poly(|\univ|, 1/\delta)$-time for all $\delta > 0$.
\end{theorem}

\begin{proof}[Proof of Theorem~\ref{thm:scaffold-approx-sub}]
Let $S$ be the independent set of $(E_{\gtest},M(\gtest))$ approximately maximizing $g$, as computed in line~\ref{al:cg}.
    Then, $\overbar{w}^\intercal \mu_S = \Delta^{-1}\cdot\sum_{k=1}^\Delta\sum_{y\in\univ}\mu_S(y)w^{(k)}(y)$.
    But, because lines 6-9 make $\sum_{y \in S(k)}w^{(k)}(y) = f(S(k))$, we find,
    \[
        \frac{1}{\Delta}\cdot\sum_{k=1}^\Delta\sum_{y\in\univ}\mu_S(y)w^{(k)}(y) \geq \frac{1}{\Delta}\cdot\sum_{k=1}^\Delta\sum_{y\in S(k)}w^{(k)}(y) \geq \frac{1}{\Delta}\cdot\sum_{k=1}^\Delta f(S(k)) = \frac{g(S)}{\Delta},
    \]
    hence $\overbar{w}^\intercal \mu_S \geq \Delta^{-1}\cdot g(S)$.

    Now, let $v^\star$ be the solution computed by Edmonds' algorithm over the base polytope of $\gtest$ when given weights $\overbar{w}$, and examine the final word $\alpha$ computed in line~\ref{al:scaf}.
    Because no coordinate in $v^\star$ can exceed value $\Delta$, the law of diminishing returns shows,
    \begin{equation}\label{eq:a}
        \sum_{y \in \widetilde{\alpha}} \overbar{w}(y) \geq \frac{\overbar{w}^\intercal v^\star}{\Delta}.
    \end{equation}
    But, the natural matroid construction shows $\mu_S$ is in the base polytope, and so by optimality of $v^\star$ we see $\overbar{w}^\intercal\mu_S \leq \overbar{w}^\intercal v^\star$.
    So, combining this with the insight of the last paragraph,
    \[
        \frac{\overbar{w}^\intercal v^\star}{\Delta} \geq \frac{\overbar{w}^\intercal \mu_S}{\Delta} \geq \frac{g(S)}{\Delta^2}.
    \]
    Then, noting that $\sum_{y \in X}w^{(k)}(y) \leq f(X)$ for all $k \in\{1,\ldots,\Delta\}$ and $X\subseteq \univ$ by the law of diminishing returns, combining the last equation with Eq.~\ref{eq:a} shows,
    \begin{equation}\label{eq:cg-key}
        f(\widetilde{\alpha}) \geq \sum_{y \in \widetilde{\alpha}} \overbar{w}(y) \geq \frac{g(S)}{\Delta^2}.
    \end{equation}
    To finish things up, let $\alpha_\mathrm{OPT}$ be an optimal basic word, and 
    $T_{\alpha_\mathrm{OPT}} = \{(y,1)\mid y\in\widetilde{\alpha}_\mathrm{OPT}\}$.
    See that,
    \[
        g(T_{\alpha_\mathrm{OPT}}) = f(\widetilde{\alpha}_\mathrm{OPT}) + \sum_{k=2}^\Delta f(\varnothing) = f(\widetilde{\alpha}_\mathrm{OPT}) = \mathrm{OPT}.
    \]
    So, combining this with the guarantees of \cite{calinescu2011maximizing} and Eq.~\ref{eq:cg-key} gives,
    \[
        \mathbb{E}[f(\widetilde{\alpha})] \geq \frac{\mathbb{E}[g(S)]}{\Delta^2} \geq \frac{1 - 1/e - \delta}{\Delta^2}\cdot\mathrm{OPT},
    \]
    which is the desired result.
\end{proof}
\begin{remark}\label{rem:imp}
    We stated our guarantees with respect to the original analysis of continuous greedy with pipage rounding from \cite{calinescu2011maximizing} because it is the most recognizable available method to a broad audience.
    However, because we used this method as a black box it can be replaced to obtain certain improvements.
    For example, \cite{chekuri2010dependent,filmus2014monotone} can obtain high probability guarantees, \cite{sviridenko2017optimal} can be used to optimally incorporate the function curvature into the guarantee, while recent work \cite{buchbinder2025deterministic} presents deterministic methods.
\end{remark}
\section{Marginal Adjusted Greedy Algorithms}\label{sec:mag}

The last algorithms were based upon ``scaffolding'' the solution with good letters from a relaxed solution, and using the fill-in procedure from Lemma~\ref{lem:projection} to project this solution to a characteristic vector corresponding to a basic word.
However, the fill-in procedure was such that the letters were chosen only for feasibility, not with any heuristic guiding their value for a solution.
In the next algorithms, we instead select letters using a greedy heuristic driven by the greatest representation.
These algorithms will lead to approximation guarantees varying in the largest matroid connectivity $\conn$, which (recall from Section~\ref{sec:intro}) can be interpreted as a distance of the MPS from a poset antimatroid.
So, the methods presented here are complementary to those in Section~\ref{sec:scaffold}.

\subsection{Additive Maximization}\label{sec:mag-am}

\begin{algorithm}[t]
    \DontPrintSemicolon
    \caption{Marginal Adjusted Greedy for Additive Maximization}\label{alg:mag}

    $\alpha \gets \epsilon$\;

    \;
    \While{$|\alpha| < r_\mathcal{M}(\univ)$}{
        $\ell \gets |\alpha|$\;
        $x_1 \ldots x_\ell \gets \alpha$\;
        \;
        $K_0 \gets \varnothing$\;

        \lFor{$i \in \{1,\ldots,\ell\}$}{$K_i \gets \cl(\{x_1, \ldots, x_i\})$\;\label{al:ki}}

        \For{$i \in \{0,\ldots,\ell\}$}
        {
            \If{$i < \ell$}
            {
                $\Psi_i \gets \{ y \in \Sigma \mid K_i \subsetneq \cl(K_i + y) \subsetneq K_{i+1}\}$\;\label{al:dif}
            }
            \Else
            {
                $\Psi_i \gets \{ y \in \Sigma \mid K_i \subsetneq \cl(K_i + y) \subseteq \univ\}$\;\label{al:dif-boundary}
            }
        }
        \;

        $(y^\star, i^\star) \gets \operatorname*{arg\,max}\limits_{y \in \Psi_i:0 \le i \le \ell}\frac{w(y)}{(\mathord{\accentset{\vee}{\rho}}/K_i)(y)}$\;\label{al:ma}

        \;

        $\alpha \gets x_1 \cdots x_{i^\star}\, y^\star\, x_{i^\star+1} \cdots x_\ell$\;
    }
\;
    \Return{$\alpha$}

\end{algorithm}

Algorithm~\ref{alg:mag} works by iteratively selecting letters to insert into a working solution (initially the empty word) in a way that both achieves good value and converges to a feasible word.
In each round, the method first computes in line~\ref{al:ki} the closed set $K_i$ given by each prefix $x_1\ldots x_i$ of the current value of $\alpha$, which we assume to be of length $\ell$.
We also make $K_0 = \varnothing$.
Then, in line~\ref{al:dif}, for each $i \in \{0,1,\ldots,\ell\}$ we compute the set $\Psi_i$ consisting of letters which strictly increase the closed set spanned by $K_i$ to a closed set strictly less than $K_{i+1}$.
The set $\Psi_\ell$ poses a boundary condition which is handled in line~\ref{al:dif-boundary}.
We save such letters in the set $\Psi_i$ because they can be inserted into $\alpha$ between $x_i$ and $x_{i+1}$ (or appended onto the beginning or end of $\alpha$ when $i$ equals 0 or $\ell$, respectively) while maintaining the invariant that the process will converge to a feasible word by way of Lemma~\ref{lem:projection}.
In particular, one can view $\alpha$ as consisting of decisions we have already committed to (along with their relative ordering).
Under this perspective, we use the closure structure and construction of the sets $\{\Psi_i\}_i$ to identify new decisions we can feasibly make while still preserving the structure of the decisions we have already committed to.
This is shown in a technical way in Eq.~\ref{eq:inv} in the proof of Theorem~\ref{thm:mag-approx}.
Finally, observe each letter is a member of at most one $\Psi_i$, because $y \in \Psi_i$ implies $y\in K_{i+1}\setminus K_i$ and the sequence $\{K_i\}_i$ forms a filtration.

Keeping this in mind, in line~\ref{al:ma} we compute a marginal-adjusted weight for every letter $y \in \bigcup_{i = 0}^\ell \Psi_i$.
Specifically, when $y\in\Psi_i$, we divide $w(y)$ by the marginal return $(\gtest/K_i)(y)$.
This is because $(\gtest/K_i)(y)$ is exactly the difference in grade between $K_i$ and $\cl(K_i + y)$, and so a large marginal return corresponds to a large jump in the closed set structure.
If we commit to selecting $y$, then we commit to selecting a sequence of letters making a saturated chain between $K_i$ and $\cl(K_i + y)$.
This is, the larger the size of this jump, the more we constrain the set of letters we can feasibly select in future rounds to insert into the solution.
So, we divide the weight by $(\gtest/K_i)(y)$ to penalize letters which induce large jumps in the closure structure.
Then, in line~\ref{al:ma}, we select the letter $y^\star$ with the largest marginal adjusted weight, and save the index $i^\star$ with $y^\star \in \Psi_{i^\star}$.
We then insert $y^\star$ into $\alpha$ between the letters $x_{i^\star}$ and $x_{i^\star + 1}$ (or place onto the beginning or end of $\alpha$ when $i$ equals 0 or $\ell$ respectively), and repeat until $|\alpha|$ equals the length of a basic word in $\mathcal{L}(\mathcal{M},\leadsto)$.

For the analysis, let $\alpha_\mathrm{OPT}$ be an optimal word. 
We need to ``charge'' each letter in the optimal word to a decision made by the algorithm.
So, first let $y^\star_t$ be the letter selected in round $t$ and $K_t^\star$ be $K_{i^\star}$ from that round.
Also, let $R = \gtest(\univ)$.
By optimism, for all letters $z\in\univ$, there exists $t \in \{1,\ldots,R\}$ such that $z\in\cl(K_t^\star + y^\star_t) \setminus K_t^\star$.
Then, let $L_t$ be the letters from $\widetilde{\alpha}_\mathrm{OPT}$ for which $t$ is the last round this occurs.
Formally, this is given by the recursion over $t\in\{1,\ldots,R\}$,
\begin{equation}\label{eq:lt}
    L_t \triangleq \big(\cl(K_t^\star + y^\star_t) \setminus K_t^\star\big)\setminus \left(\bigcup_{t' = t+1}^{R} L_t\right),\quad\text{with }L_{R} = \cl(K_R^\star + y^\star_R) \setminus K_R^\star.
\end{equation}
And, by our prior observation we see $(L_t)_t$ is a partition of $\widetilde{\alpha}_\mathrm{OPT}$.
This sets up the sum,
\begin{equation}\label{eq:mag-imp1}
    \sum_{y \in\widetilde{\alpha}_\text{OPT}}w(y) = \sum_{t = 1}^{R} \left(\sum_{y \in L_t}w(y)\right).
\end{equation}
But, letting $\delta_t = (\gtest/K_t^\star)(y^\star_t)$, the algorithm's choice rule makes,
\[
     w(y) \leq \frac{w(y^\star_t)}{\delta_t}\cdot \left(\gtest/K_{i^\star}\right)(y),\quad\forall y\in L_t,
\]
while $L_t\subseteq \cl(K_t^\star + y^\star_t) \setminus K_t^\star$ makes $\left(\gtest/K_t^\star\right)(y) \leq \left(\gtest/K_t^\star\right)(y^\star_t) = \delta_t$ for all $y \in L_t$.
Therefore,
\begin{equation}\label{eq:mag-imp2}
    \sum_{t = 1}^{R} \left(\sum_{y \in L_t}w(y)\right) \leq \sum_{t = 1}^{R}\left( |L_t|\cdot \frac{w(y^\star_t)}{\delta_t}\cdot \left(\gtest/K_{i^\star}\right)(y) \right)\leq \sum_{t = 1}^{R} |L_t|\cdot w(y^\star_t).
\end{equation}
So, noting that $\sum_{t = 1}^{R} w(y^\star_t)$ is the value of the solution returned by Algorithm~\ref{alg:mag}, we need only obtain a uniform bound on $|L_t|$ for an approximation guarantee.
This is our approach.
Recall from Section~\ref{sec:intro} that $\conn$ is the maximum value of the matroid connectivity function of the matroid defining the MPS (see Eq.~\ref{eq:conn-def}).
A careful analysis will show $|L_t| \leq 1 + \conn$.

\begin{theorem}\label{thm:mag-approx}
    Let $\mathcal{L}(\mathcal{M},\leadsto)$ be a matroidal prerequisite system, $w:\univ\to\mathbb{R}_+$ a nonnegative weight function,
    and $\mathrm{OPT} =\max_{\alpha \in \mathcal{L}(\mathcal{M},\leadsto)} \sum_{y \in \widetilde{\alpha}} w(y)$.
    Algorithm~\ref{alg:mag} outputs $\alpha\in\mathcal{L}(\mathcal{M},\leadsto)$ with,
    \[
        \sum_{y \in \widetilde{\alpha}}w(y) \geq \frac{1}{1+\conn}\cdot\mathrm{OPT}.
    \]
\end{theorem}
\begin{proof}
    Firstly, see that $\alpha$ is feasible because lines~\ref{al:dif} and~\ref{al:dif-boundary} in Algorithm~\ref{alg:mag} ensure the following invariant: the evaluations of $\cl$ on each prefix of the working solution forms a strictly increasing chain in $(\kappa(\glang)/\mathord{\sim})$.
    Since $\sigma_{\gtest} = \cl$, it can be seen then that the working solution satisfies the premise of Lemma~\ref{lem:projection} at all times.
    Hence, once it achieves the length of a basic word, it must be feasible.

    Now, before continuing, observe the letters of the solution $\alpha$ are not arranged in the order in which they were inserted.
    In light of this, let $\pi:\{1,\ldots,R\}\to\{1,\ldots,R\}$ map the index corresponding to the position of a letter in $\alpha$ to the round index it was added to the solution.
    Then, we may write $\alpha = y^\star_{{\pi}(1)}\ldots y^\star_{{\pi}(R)}$ while still identifying the round a letter was added.
    Furthermore, also define,
    $$X_i \triangleq \left\{y^\star_{{\pi}(1)},\ldots, y^\star_{{\pi}(i)}\right\},$$
    as the first $i$ letters in the algorithm's solution $\alpha$.
    Take note that as $y^\star_{{\pi}(1)},\ldots, y^\star_{{\pi}(0)} = \epsilon$ by convention, it follows that $X_0 = \varnothing$.
    Moreover, the insertion scheme of Algorithm~\ref{alg:mag} maintains,
    \begin{equation}\label{eq:inv}
        \cl(X_i) = \cl\left(K^\star_{\pi(i)} + y^\star_{\pi(i)}\right),\quad\forall i \in \{1,\ldots,R\}.
    \end{equation}
    To see why, fix $i \in \{1,\ldots, R\}$ and select an index $j$ minimizing $\pi(j)$ among $j < i$ making $\pi(j) > \pi(i)$; this is, $y^\star_{\pi(j)}$ is the first letter preceding $y^\star_{\pi(i)}$ in position in $\alpha$ which was selected for insertion in a later round than $y^\star_{\pi(i)}$ was selected.
    Note that if no such $j$ exists, then Eq.~\ref{eq:inv} holds as $K^\star_{\pi(i)} = \cl(X_{i-1})$ would have to be true.
    Then, $y^\star_{\pi(j)}$ belongs to a set in the sequence $(\Psi_{i'})_{i'}$ computed in line~\ref{al:dif}, and because $y^\star_{\pi(j)}$ precedes $y^\star_{\pi(i)}$ in position in $\alpha$ and minimizes $\pi(j)$ amongst such letters, this set from $(\Psi_{i'})_{i'}$ must be a subset of,
    \begin{equation}\label{eq:4545}
        \cl\left(\left\{y^\star_{\pi(k)}\bigm\vert \pi(k) \leq \pi(i)\text{ and }k \leq i\right\}\right) = \cl\left(K^\star_{\pi(i)} + y^\star_{\pi(i)}\right).
    \end{equation}
    So, $y_{\pi(j)} \in \cl\left(K^\star_{\pi(i)} + y^\star_{\pi(i)}\right)$.
    Now, to induct, select an arbitrary index $j'$ with $j' < i$ and $\pi(j') > \pi(j)$.
    This is, $y^\star_{\pi(j')}$ is any letter preceding $y^\star_{\pi(i)}$ in position in $\alpha$ which was selected for insertion in a later round than $y^\star_{\pi(i)}$ and $y^\star_{\pi(j)}$.
    By inductive hypothesis,
    \[
        k < i\text{ and } \pi(i) < \pi(k) < \pi(j') \implies y^\star_{\pi(k)} \in \cl\left(K^\star_{\pi(i)} + y^\star_{\pi(i)}\right).
    \]
    Combining this with Eq.~\ref{eq:4545} shows,
    \[
        \cl\left(\left\{y^\star_{\pi(k)}\bigm\vert \pi(k) < \pi(j')\text{ and }k \leq i\right\}\right) = \cl\left(K^\star_{\pi(i)} + y^\star_{\pi(i)}\right).
    \]
    The left hand side is the closure of the support of the prefix of the working solution ending at $y^\star_{\pi(i)}$ in round $\pi(j')$.
    Then, because $j' < i$, i.e. $y^\star_{\pi(j')}$ precedes $y^\star_{\pi(i)}$ in position in $\alpha$, line~\ref{al:dif} implies $y^\star_{\pi(j')} \in \cl\left(K^\star_{\pi(i)} + y^\star_{\pi(i)}\right)$ as well.
    So, by the principle of induction we see $X_{i-1} \subseteq \cl\left(K^\star_{\pi(i)} + y^\star_{\pi(i)}\right)$, which by $X_{i-1} + y^\star_{\pi(i)} = X_i$ and idempotency of closure operators implies $\cl(X_{i}) \subseteq \cl\left(K^\star_{\pi(i)} + y^\star_{\pi(i)}\right)$.
    Of course, the reverse inclusion follows from,
    \begin{equation*}
        \begin{split}
            \cl\left(K^\star_{\pi(i)} + y^\star_{\pi(i)}\right) &= \cl\left(\cl\left(\left\{y^\star_{\pi(k)}\bigm\vert k < i \text{ and }\pi(k) < \pi(i)\right\}\right) + y^\star_{\pi(i)}\right),\\
                                                                &=\cl\left(\left\{y^\star_{\pi(k)}\bigm\vert k < i \text{ and }\pi(k) < \pi(i)\right\} + y^\star_{\pi(i)}\right),
    \end{split}
    \end{equation*}
    which is a subset of $\cl(X_i)$ by the monotonicity of closure operators.
    And so, Eq.~\ref{eq:inv} follows.

    With that done, the approximation guarantee will follow by proving the following key insight,
    \begin{equation}\label{eq:key-mag}
        z \in L_{{\pi}(i)} \implies z \in \cl\left(X_i\right)\setminus\cl\left(X_{i-1}\right).
    \end{equation}
    To see why, observe $L_{{\pi}(i)} \subseteq \cl\left(K^\star_{{\pi}(i)} + y^\star_{{\pi}(i)}\right)\setminus K^\star_{{\pi}(i)}$.
    So, by way of contradiction, suppose Eq.~\ref{eq:key-mag} does not hold.
    There exists $z\in L_{{\pi}(i)}$ with $z \notin \cl\left(X_i\right)\setminus\cl\left(X_{i-1}\right)$.
    As $z \in \cl\left(K^\star_{{\pi}(i)} + y^\star_{{\pi}(i)}\right)\setminus K^\star_{{\pi}(i)}$ and $\cl\left(K^\star_{{\pi}(i)} + y^\star_{{\pi}(i)}\right) \subseteq \cl(X_i)$, since $K^\star_{{\pi}(i)} + y^\star_{{\pi}(i)}$ is the union of the $\cl$-closure of a subset of $X_{i-1}$ unioned with the $i^\text{th}$ letter in $\alpha$, for Eq.~\ref{eq:key-mag} to not hold it must be that $z \in \cl(X_{i-1})$.
    This shows,
    \begin{equation}\label{eq:mag-contra}
        \left(\exists z \in L_{{\pi}(i)}\right)\;z \in \cl\left(X_{i-1}\right)\setminus K^\star_{{\pi}(i)}.
    \end{equation}
    However, $\cl(X_{i-1}) = \cl\left(K^\star_{\pi(i-1)} + y^\star_{\pi(i-1)}\right)$ by Eq.~\ref{eq:inv}, so by $z \in L_{{\pi}(i)}$ it must be that ${\pi}(i-1) < {\pi}(i)$.
    This means that $y^\star_{{\pi}(i-1)}$ was in the support of the working solution when $y^\star_{\pi}(i)$ was inserted.
    Since $y^\star_{{\pi}(i-1)}$ directly precedes $y^\star_{\pi}(i)$ in position in $\alpha$, this implies,
    $$K^\star_{{\pi}(i)} = \cl\left(K^\star_{\pi(i-1)} + y^\star_{\pi(i-1)}\right).$$
    But $\cl(X_{i-1}) = \cl\left(K^\star_{\pi(i-1)} + y^\star_{\pi(i-1)}\right)$ implies $\cl\left(X_{i-1}\right) = K^\star_{{\pi}(i)}$,
    which contradicts Eq.~\ref{eq:mag-contra}.
    So, we find Eq.~\ref{eq:key-mag} holds true.

    As a result, we observe $L_t \subseteq \widetilde{\alpha}_\mathrm{OPT} \cap \left(\cl\left(X_i\right)\setminus\cl\left(X_{i-1}\right)\right)$.
    Now, we make use of the matroid structure to introduce $\conn$.
    Because $\widetilde{\alpha}_\text{OPT}$ is a basic set in $\mathcal{M}$ (see Section~\ref{sec:mps}), 
    \begin{equation}\label{eq:conn-intro}|\widetilde{\alpha}_\text{OPT} \cap \cl\left(X_{i-1}\right)|=r_\mathcal{M}(\Sigma) - r_\mathcal{M}(\widetilde{\alpha}_\text{OPT} \setminus \cl\left(X_{i-1}\right)).\end{equation}
    Continuing on, note that because $\cl(X_{i})$ is the smallest flat support containing $y^\star_1\ldots y^\star_i$, it follows that $\cl(X_{i}) = \kappa[y^\star_1\ldots y^\star_{i}]$.
    This logic also shows $\cl(X_{i-1}) = \kappa[y^\star_1\ldots y^\star_{i-1}]$.
    Then, $\cl(X_i) = \sigma_\mathcal{M}(\{y^\star_1,\ldots,y^\star_{i}\})$ and $\cl(X_{i-1}) = \sigma_\mathcal{M}(\{y^\star_1,\ldots,y^\star_{i-1}\})$ by Lemma~\ref{lem:int}.1.
    Because the feasible words of the MPS are independent in $\mathcal{M}$ it follows that,
    $$r_\mathcal{M}(\cl(X_i)) = r_\mathcal{M}(\cl(X_{i-1})) + 1.$$
    Since independence implies $|\widetilde{\alpha}_\mathrm{OPT} \cap \cl(X_i)| \leq r_\mathcal{M}(\cl(X_{i}))$, we see $|\widetilde{\alpha}_\mathrm{OPT} \cap \cl(X_i)| \leq r_\mathcal{M}(\cl(X_{i-1})) + 1$.
    This observation combines with Eq.~\ref{eq:conn-intro} for,
    \begin{equation*}
        \begin{split}
            |L_t| &\leq |\widetilde{\alpha}_\mathrm{OPT} \cap \left(\cl(X_{i})\setminus\cl(X_{i-1})\right)|,\\
                  &= |\widetilde{\alpha}_\mathrm{OPT} \cap\cl(X_i)| - |\widetilde{\alpha}_\mathrm{OPT} \cap\cl(X_{i-1})|,\\
                  &\leq 1 + r_\mathcal{M}(\cl(X_{i-1})) + r_\mathcal{M}(\widetilde{\alpha}_\text{OPT} \setminus \cl(X_{i-1})) - r_\mathcal{M}(\univ),\\
                  &\leq 1 + \left(r_\mathcal{M}(\cl(X_{i-1})) + r_\mathcal{M}(\univ \setminus \cl(X_{i-1})) - r_\mathcal{M}(\univ)\right), 
        \end{split}
    \end{equation*}
    which is less than or equal to $1 + \conn$ by definition (see Eq.~\ref{eq:conn-def}).
    And so, by Eq.s~\ref{eq:mag-imp1} and~\ref{eq:mag-imp2},
    \[
        \mathrm{OPT} \leq \sum_{t = 1}^{R} |L_t|\cdot w(y^\star_t) \leq (1 + \conn)\cdot\sum_{t = 1}^R w(y^\star_t),
    \]
    which shows $\alpha$ is a $(1+\conn)$-approximation of the optimal solution.
\end{proof}

\subsection{Monotone Submodular Maximization}\label{subsec:mag-sub}

\begin{algorithm}[t]
    \DontPrintSemicolon
    \caption{Marginal Adjusted Greedy for Submodular Maximization}\label{alg:mag-sub}

    $\alpha \gets \epsilon$\;

    \;
    \While{$|\alpha| < r_\mathcal{M}(\univ)$}{
        $\ell \gets |\alpha|$\;
        $x_1 \ldots x_\ell \gets \alpha$\;
        \;
        $K_0 \gets \varnothing$\;
        \lFor{$i \in \{1,\ldots,\ell\}$}{$K_i \gets \cl(\{x_1, \ldots, x_i\})$}
        \;

        \For{$i \in \{0,\ldots,\ell\}$}
        {
            \If{$i < \ell$}
            {
                $\Psi_i \gets \{ y \in \Sigma \mid K_i \subsetneq \cl(K_i + y) \subsetneq K_{i+1}\}$\;\label{al:dif2}
            }
            \Else
            {
                $\Psi_i \gets \{ y \in \Sigma \mid K_i \subsetneq \cl(K_i + y) \subseteq \univ\}$\;\label{al:dif-b2}
            }
        }
        \;

        $(y^\star, i^\star) \gets \operatorname*{arg\,max}\limits_{y \in \Psi_i:0 \le i \le \ell}\frac{(f/\widetilde{\alpha})(y)}{(\mathord{\accentset{\vee}{\rho}}/K_i)(y)}$\;\label{al:saw}

        \;

        $\alpha \gets x_1 \cdots x_{i^\star}\, y^\star\, x_{i^\star+1} \cdots x_\ell$\;
    }
\;
    \Return{$\alpha$}

\end{algorithm}

We can extend the method of Algorithm~\ref{alg:mag} to the case of submodular maximization by modifying the selection heuristic.
This is shown in Algorithm~\ref{alg:mag-sub}, where the only change from Algorithm~\ref{alg:mag} is that we compute the adjusted weight of a letter $y$ in line~\ref{al:saw} using the marginal return $(f/\widetilde{\alpha})(y)$ under $f$ against the current working solution.
We reuse parts of the analysis of the previous section to obtain the following.

\begin{theorem}\label{thm:mag-approx-sub}
    Let $\mathcal{L}(\mathcal{M},\leadsto)$ be a matroidal prerequisite system, $f:2^\univ\to\mathbb{R}$ a (normalized) monotone submodular function, and $\mathrm{OPT} = \max_{\alpha \in \mathcal{L}(\mathcal{M},\leadsto)}f(\widetilde{\alpha}).$
    Algorithm~\ref{alg:mag-sub} outputs $\alpha$ with,
    \[
        f(\widetilde{\alpha}) \geq \frac{1}{2 + \conn}\cdot\mathrm{OPT}.
    \]
\end{theorem}
\begin{proof}
    Our analysis reuses parts of that presented for Algorithm~\ref{alg:mag}.
    Notably, as lines~\ref{al:dif2} and~\ref{al:dif-b2} remain unchanged from Algorithm~\ref{alg:mag}, it follows that $\alpha$ is feasible by way of Lemma~\ref{lem:projection}.

    Continuing on, define $L_t$ in the same way as Eq.~\ref{eq:lt}.
    Notice that the argument bounding $|L_t|$ presented in the proof of Theorem~\ref{thm:mag-approx} only uses properties of the MPS and not the additive objective function, so we have $|L_t| \leq 1 + \conn$.
    Also let $y^\star_t$ be the letter selected in round $t$, $\alpha_\mathrm{OPT}$ and $K^\star_t$ defined as before, as well as $\alpha^{(t)}$ be the value of the working solution $\alpha$ at the end of round $t$ (with the convention $\alpha^{(0)} = \epsilon$).
    Now, our choice rule makes,
    \[
        \left(f/\widetilde{\alpha}^{(t-1)}\right)(z) \leq \frac{\left(f/\widetilde{\alpha}^{(t-1)}\right)(y^\star_t)}{(\gtest/K^\star_t)(z)}(\gtest/K^\star_t)(y^\star_t) \leq \left(f/\widetilde{\alpha}^{(t-1)}\right)(y^\star_t),\quad\forall z \in L_t,
    \]
    and so we find that for all $t \in \{1,\ldots,R\}$,
    \[
        \sum_{z\in L_t}\left(f/\widetilde{\alpha}^{(t-1)}\right)(z) \leq |L_t| \cdot \left(f/\widetilde{\alpha}^{(t-1)}\right)(y^\star_t) \leq (1+\conn)\cdot\left(f/\widetilde{\alpha}^{(t-1)}\right)(y^\star_t).
    \]
    So, because $f(\widetilde{\alpha}) = \sum_{t = 1}^R \left(f/\widetilde{\alpha}^{(t-1)}\right)(y^\star_t)$, summing over all rounds shows,
    \begin{equation}\label{eq:conn-f}
        \sum_{t = 1}^R \sum_{z\in L_t}\left(f/\widetilde{\alpha}^{(t-1)}\right)(z) \leq (1 + \conn)\cdot f(\widetilde{\alpha}).
    \end{equation}
    We relate the left-hand side to $\mathrm{OPT}$ via,
    \[
        \mathrm{OPT} \leq f(\widetilde{\alpha} \cup \widetilde{\alpha}_\mathrm{OPT}) \leq f(\widetilde{\alpha}) + \sum_{z \in \widetilde{\alpha}_\mathrm{OPT}\setminus\widetilde{\alpha}} (f/\widetilde{\alpha})(z),
    \]
    as the first inequality follows by monotonicity, while the second follows by the law of diminishing returns.
    Well, by the law of diminishing returns we also have,
    \[
        z \in L_t \implies \left(f/\widetilde{\alpha}^{(t-1)}\right)(z) \geq (f/\widetilde{\alpha})(z),
    \]
    since $\widetilde{\alpha}^{(t-1)}\subseteq \widetilde{\alpha}$ always holds by our construction.
    This and Eq.~\ref{eq:conn-f} implies,
    \[
        \mathrm{OPT} \leq f(\widetilde{\alpha}) + \sum_{t = 1}^R \sum_{z \in L_t} \left(f/\widetilde{\alpha}^{(t-1)}\right)(z) \leq (2 + \conn)\cdot f(\widetilde{\alpha}).
    \]
    Rearranging terms gives the desired result.
\end{proof}

\section{Hardness of Approximation of Additive Maximization}\label{sec:hardness}

The goal of this section is showing that $\min\{\Delta,\conn\}^{o(1)}$-approximations for nonnegative additive maximization over an MPS is intractable.
We reduce from \emph{Densest $k$-Subgraph (D$k$-S)}:
\begin{quote}
    \emph{Given an undirected graph $(V,E)$ and integer $k$, select a vertex subset $S\subseteq V$ of size $|S|=k$ maximizing the induced edges.}
\end{quote}
The Gap Exponential Time Hypothesis (Gap-ETH) conjectures the existence of a small $\delta > 0$ for which no subexponential time algorithm can distinguish between a satisfiable 3-SAT formula and one where at most a $(1-\delta)$-fraction of clauses can be satisfied \cite{manurangsi2017birthday,dinur2016mildly,chalermsook2017gap}.
The Gap-ETH is now a standard workhorse for proving fine-grained inapproximability results.
And, under Gap-ETH there is no $|V|^{o(1)}$-approximation to D$k$-S \cite{manurangsi2017almost}.

For a coming arithmetic exercise, define the constant $B = |E| + 1$.
We now construct an MPS whose nonnegative additive maximization solutions encode D$k$-S solutions in a way preserving their approximation quality. 
For the alphabet, for each $u \in V$ make $B$ ``dummy'' letters $d_{u, 1},\ldots, d_{u, B}$, include each vertex $u$, and include each edge $e = \{u, v\} \in E$.
For the weights, set,
\begin{equation}\label{eq:red-weights}
    w(y) =
    \begin{cases}
        1,\quad&\text{if $y$ is an edge letter}, \\
        0,\quad&\text{otherwise}.
    \end{cases}
\end{equation}
And, let the underlying matroid $\mathcal{M}$ be the uniform matroid of rank $R \triangleq k(B+1) + |E| + 1$.
Finally, for the prerequisite order, each vertex $u$ is greater than all its dummy letters $d_{u, 1},\ldots, d_{u, B}$ and each edge $e = \{u, v\} \in E$ is greater than $u$ and $v$.
Recall, we showed in Section~\ref{sec:mps} that uniform matroids always satisfy compatibility with $(\univ,\leadsto)$.
So, $\mathcal{L}(\mathcal{M},\leadsto)$ is a well-defined MPS.

To begin, we describe how the solutions values of the MPS and Dk-S instances are related in the following two claims.
This will then be used to prove Lemma~\ref{lem:ap-reduction}, which gives an approximation-preserving reduction from Dk-S to nonnegative additive maximization over an MPS.
\begin{claim}
    There is an efficient procedure to extract a $k$-subgraph with at least $\mathrm{VAL}-1$ edges from every basic word of cumulative weight $\mathrm{VAL}$.
\end{claim}
\begin{proof}
    Let $\alpha \in \mathcal{L}(\mathcal{M},\leadsto)$ be a basic word of value $\mathrm{VAL}$, and let $\alpha'$ be the prefix excluding the last letter.
    We exclude the last letter because it can correspond to an edge which was added without both its incident vertices being present in the word.
    However, because the matroid is uniform, it is not possible for the prerequisites of any edge in $\widetilde{\alpha}'$ to be satisfied ``virtually'' without their membership in the support.
    Therefore, because the edges correspond to the only letters with nonzero weight, if we let the cumulative weight of the support of $\alpha'$ be $\mathrm{VAL}'$ then the vertices in $\widetilde{\alpha}'$ must induce a subgraph with at least $\mathrm{VAL}'$ edges.
    
    Then, examine the set $X = \{v \in V\mid (\exists e \in \widetilde{\alpha}' \cap E)\;v \in e\}$.
    We claim $|X| \leq k$.
    By way of contradiction, suppose this is not true.
    Our prior observations make $X\subseteq\widetilde{\alpha}'$.
    Note then that for every letter in $X$, the $B$ dummy prerequisite letters must also be present.
    So, if $|X| > k+1$ it follows that there are at least $k+1$ vertex letters and $B(k+1)$ dummy letters in $\widetilde{\alpha}'$.
    By recalling $B = |E| + 1$,
    \[
        |\widetilde{\alpha}'| \geq (k+1)(B+1) = k(B+1) + |E| + 2 > R.
    \]
    This means $\widetilde{\alpha}'$ is not independent, which contradicts the fact that it is the prefix of a basic word of $\mathcal{L}(\mathcal{M},\leadsto)$.
    And so, because $|X| \leq k$, one can output a feasible D$k$-S solution of value at least $\mathrm{VAL}'$ by listing the vertices in $X$ along with a choice of $k -|X|$ additional vertices.
    Since $\mathrm{VAL}' \geq \mathrm{VAL} - 1$, this shows the claim.
\end{proof}

\begin{claim}
    For any D$k$-S solution inducing $\mathrm{VAL}$ edges, there exists a basic word of cumulative weight at least $\mathrm{VAL}$.
\end{claim}
\begin{proof} 
    Let $X$ be a D$k$-S solution of value $\mathrm{VAL}$. 
    Take $S = \{(u,v) \in E\mid \{u,v\}\subseteq X\}$ as the edges in the induced subgraph; clearly $|S| =\mathrm{VAL}$.
    Now, examine a word $\alpha$ formed by taking any linear extension of the ideal generated by $S$ in $(\univ,\leadsto)$.
    Observe, $|X|=k$ implies there can be at most $k$ vertex letters in this ideal, and each vertex letter is above exactly $B$ letters associated with its dummies.
    Thus,
    $$
        |\alpha| \leq k(B+1) + (\mathrm{VAL}) \leq k(B+1) + |E|,
    $$
    and so $\widetilde{\alpha}$ is independent.
    Furthermore, the prerequisite conditions are trivially satisfied by $\alpha$ since it is a linear extension of an ideal (i.e. for any $y\in\widetilde{\alpha}$, \emph{every} $x\leadsto y$ is such that $x \in \widetilde{\alpha}$ in a position occurring before the appearance of $y$), so $\alpha \in \mathcal{L}(\mathcal{M},\leadsto)$.
    Since the support of $\alpha$ contains $\mathrm{VAL}$ edges, its cumulative weight is at least $\mathrm{VAL}$.
    Therefore, any basic word formed by concatenating $\alpha$ with a suffix from $\mathcal{L}(\mathcal{M},\leadsto)/[\alpha]$ gives the result.
\end{proof}

\begin{lemma}\label{lem:ap-reduction}
    Fix a D$k$-S instance.
    Let $\univ$, $\mathcal{M}$, and $(\univ,\leadsto)$ as constructed, and $w:\univ\to\mathbb{R}$ be as defined in Eq.~\ref{eq:red-weights}.
    Let $\mathrm{OPT}_\text{D$k$-S}$ be the number of edges in an optimal D$k$-S solution, and $\mathrm{OPT}_\text{MPS}=\max_{\alpha \in \mathcal{L}(\mathcal{M},\leadsto)} \sum_{y\in\widetilde{\alpha}} w(y)$.
    If there exists an efficient $\delta$-approximation algorithm to nonnegative additive maximization over the MPS, then there exists an efficient algorithm giving a D$k$-S solution inducing $\delta\cdot(\mathrm{OPT}_\text{D$k$-S} + O(1)) - 1$ edges.
\end{lemma}
\begin{proof}
    Let $\mathrm{OPT}_\text{MPS}$ be the weight of an optimal basic word.
    The first claim we proved shows $\mathrm{OPT}_\text{MPS} \leq \mathrm{OPT}_\text{D$k$-S} + 1$ while the second shows $\mathrm{OPT}_\text{D$k$-S} \leq \mathrm{OPT}_\text{MPS}$.
    So, $\mathrm{OPT}_\text{MPS}=\mathrm{OPT}_\text{D$k$-S} + O(1)$.
    Moreover, the first claim also shows an efficient procedure for extracting a D$k$-S solution of value $\delta\cdot\mathrm{OPT}_\text{MPS} - 1$ from any $\delta$-approximation to the nonnegative additive maximization over the MPS.
    As our observations make,
    \[
        \delta\cdot\mathrm{OPT}_\text{MPS} - 1 = \delta\cdot(\mathrm{OPT}_\text{D$k$-S} + O(1)) - 1,
    \]
    we see the result.
\end{proof}

\begin{proof}[Proof of Main Result~\ref{mr:complexity}]
Now, observe the principal ideal of any edge letter is a maximum sized such ideal with cardinality exactly $2\cdot(B+1) + 1$.
Since $R \geq 2\cdot(B+1)$ for $k\geq 1$, this means for any choice of edge $e$,
$$\Delta = r_\mathcal{M}(\ideal e) = \Theta(|E|) =  O(|V|^2).$$
Therefore, by Lemma~\ref{lem:ap-reduction}, any $\Delta^{o(1)}$-approximation algorithm for an MPS computes a $|V|^{o(1)}$-approximation to D$k$-S.
Moreover, one can observe that $\conn \leq r_\mathcal{M}(\univ)$ is always true.
And so, $\conn \leq R = O(|E|) = O(|V|^2)$.
Thus, Lemma~\ref{lem:ap-reduction} also implies any $\conn^{o(1)}$-approximation algorithm for an MPS is a $|V|^{o(1)}$-approximation to D$k$-S.
Hence, $\min\{\Delta,\conn\}^{o(1)}$-approximations are intractable.
\end{proof}

\begin{remark}\label{rem:eth}
    If one only assumes ETH, then \cite{manurangsi2017almost} shows a $(|V|^{1/(\log\log |V|)^c)})$-approximation to Dk-S in polynomial-time is not possible, for some universal constant $c > 0$.
    The preceding argument then shows an efficient $(\min\{\Delta,\conn\}^{1/(\log\log \min\{\Delta,\conn\}^{c'})})$-approximation is not possible, for a universal $c'> 0$, under this weaker hypothesis.
\end{remark}

\printbibliography

\pagebreak

\appendix

\section{The Greedy Algorithm in an MPS}\label{app:greedy-fail}

Due to the proximity to matroids (and its classification as a greedoid), here we discuss the performance of the greedy algorithm on an MPS in greater detail in the context of nonnegative additive maximization.
If the matroid rank is $R$, it is possible to guarantee a $R$-approximation using the greedy algorithm.
This is because, by optimism, eventually a letter of maximal weight (over the entire alphabet) will be a continuation of the working solution built by the greedy algorithm.
We give an example showing that this is nearly tight (observe that our structurally dependent guarantees in terms of $\Delta$ and $\conn$ are significantly finer).
For this, we let $\mathrm{GREEDY}$ be the value obtained by the greedy algorithm and $\mathrm{OPT}$ the optimal value.

Let $\univ = \{q, p_1, \ldots, p_{R-1}, d_1, \ldots, d_{R-1}\}$.
Call $\{p_i\}_i$ the \emph{profit} letters and $\{d_i\}_i$ the \emph{decoy} letters.
Notice that $q$ is neither a decoy or a profit letter.
Then, make the prerequisite order $(\univ,\leadsto)$ the transitive closure of,
$$
    q\prec p_1 \prec p_2 \prec \ldots \prec p_{R-1}.
$$
Note decoys make an antichain (set of incomparable poset elements) incomparable to any profit letter.
The matroid is the uniform $\mathcal{U}^{R}_{2R - 1}$.
With $P \gg 1$, set,
\begin{equation}
    w(y) =
    \begin{cases}
        1,\quad&\text{if $y$ is a decoy}, \\
        P,\quad&\text{if $y$ is a profit letter}, \\
        0,\quad&\text{otherwise}.
    \end{cases}
\end{equation}
Observe, the letter $q$ gates the greedy algorithm from selecting any profit letters until its solution is length one less than the rank of a basic set in the matroid.
So, the greedy algorithm makes a word like $d_1 \ldots d_{R-1} p_i$ and achieves value $(R-1) + P$.
But, the optimal word is $q p_1 \ldots p_{R-1}$ and achieves value $(R-1)\cdot P$.
Then,
$$\frac{\mathrm{GREEDY}}{\mathrm{OPT}} = \frac{(R-1) + P}{(R-1)\cdot P}.$$
In particular, the ratio $\frac{\mathrm{GREEDY}}{\mathrm{OPT}}\rightarrow 1/(R-1)$ as $P$ tends to the limit, and so it is not possible to guarantee the greedy algorithm will obtain a $(R-(1+\delta))$-approximation for any $\delta > 0$.

\section{Verifying Figure~\ref{fig:bridge} Satisfies the Compatibility Condition}\label{app:ver}

Let $\mathcal{M}$ and $(\univ,\leadsto)$ be defined as in Figure~\ref{fig:bridge}.
Select $\alpha\in\mathcal{L}(\mathcal{M},\leadsto)$ and select $z \in \sigma_\mathcal{M}(\widetilde{\alpha})$.
By way of contradiction, suppose $y \leadsto z$ and $y \notin\sigma_\mathcal{M}$.
From Eq.~\ref{eq:feasibility} it follows that $z \notin \widetilde{\alpha}$, while $z \in \sigma_\mathcal{M}(\widetilde{\alpha})$ shows $\alpha \neq \epsilon$.
Moreover, $y\notin\sigma_\mathcal{M}(\widetilde{\alpha})$ implies $\alpha$ is not a basic word.
Because $t_1$ and $t_2$ are parallel in the matroid and precede all other letters in $(\univ,\leadsto)$, we see that $\{t_1,t_2\}$ is a subset of $\sigma_\mathcal{M}(\widetilde{\alpha})$ since $\alpha$ is not empty.
So, $y \notin\{t_1,t_2\}$.
Since the anomaly detection sensors $\{s_1,s_2,s_3,s_4\}$ are maximal in $(\univ,\leadsto)$, $y\leadsto z$ implies $y\notin \{s_1,s_2,s_3,s_4\}$.
Therefore, one of $y=a_1$ or $y=a_2$ holds.
Since $\alpha$ is not basic, yet $\{a_1,a_2\}\cap(\sigma_\mathcal{M}(\widetilde{\alpha})\setminus \widetilde{\alpha})$ is nonempty, it follows that $|\alpha| = 3$ with its support containing exactly one sync letter, one baseline sensor letter, and one anomaly detection letter from the block corresponding to the baseline sensor letter.
It is not possible to span $a_2$ in the linear matroid by $\{s_1,a_1\}$ or $\{s_2,a_1\}$.
We also see that $a_1$ is not spanned in the linear matroid by $\{s_3, a_2\}$, but it is spanned by $\{s_4,a_2\}$.
Therefore, one of $\alpha = t_1a_2s_4$ or $\alpha = t_2a_2s_4$ are true.
However, the only letters which $a_1$ is a prerequisite of in $(\univ,\leadsto)$ are $s_1$ and $s_2$, so $z \in \{s_1,s_2\}$.
Yet, $\{s_1,s_2\}\cap\sigma_\mathcal{M}(\widetilde{\alpha}) = \varnothing$ since $\{a_2,s_4\}$ does not span $s_1$ or $s_2$ in the linear matroid.
Therefore $z \notin \sigma_\mathcal{M}(\widetilde{\alpha})$, a contradiction.\hfill\qed

\end{document}